\documentclass[a4paper,UKenglish,cleveref,autoref]{lipics-v2019}
\usepackage[basic]{complexity}
\usepackage{todonotes}
\usepackage{tikz}
\usepackage{xcolor}
\usetikzlibrary{intersections}
\usetikzlibrary{decorations.markings}
\usetikzlibrary{arrows, automata, positioning,calc}
\usetikzlibrary{decorations.pathmorphing}
\usetikzlibrary{decorations.pathreplacing}
\usetikzlibrary{shapes.symbols,shapes.callouts,patterns}
\usetikzlibrary{math}

\def\abs#1{\ensuremath{\lvert #1\rvert}}

\newcommand{\VASS}{\mathcal{V}}

\newcommand{\Conf}{\mathit{Conf}}

\newcommand{\blocked}{\mathrm{blocked}}
\newcommand{\weight}{\mathrm{weight}}
\newcommand{\pmin}{\mathrm{pmin}}

\renewcommand{\poly}{\mathrm{poly}}

\newcommand{\set}[1]{{\{ #1 \}}}
\newcommand{\NN}{{\mathbb{N}}}

\newcommand{\runto}[1]{\stackrel{#1}{\rightarrow}}

\title{Coverability in 1-VASS with Disequality Tests} 

\titlerunning{Coverability in 1-VASS with Disequality Tests}

\author{Shaull Almagor}{Technion, Israel}{}{0000-0001-9021-1175}{has received funding from the European Union's Horizon 2020 research and innovation programme under the Marie Sk{\l}odowska-Curie grant agreement No 837327.}
\author{Nathann Cohen}{CNRS \& LRI, France}{}{}{}
\author{Guillermo A. P\'erez}{University of Antwerp,
Belgium}{}{0000-0002-1200-4952}{}
\author{Mahsa Shirmohammadi}{CNRS \& IRIF, Universit\'e de Paris, France}{}{}{}
\author{James Worrell}{University of Oxford, UK}{}{}{Supported by EPSRC Fellowship EP/N008197/1.}

\authorrunning{S.~Almagor, N.~Cohen, G.~A.~P\'erez, M.~Shirmohammadi and J.~Worrell}

\Copyright{Shaull Almagor, Nathann Cohen, Guillermo A. P\'erez,
Mahsa Shirmohammadi, James Worrell}



\acknowledgements{We thank P. Offtermatt for pointing us to
literature on NC-algorithms.}

\nolinenumbers 

\hideLIPIcs  

\EventEditors{John Q. Open and Joan R. Access}
\EventNoEds{2}
\EventLongTitle{42nd Conference on Very Important Topics (CVIT 2016)}
\EventShortTitle{CVIT 2016}
\EventAcronym{CVIT}
\EventYear{2016}
\EventDate{December 24--27, 2016}
\EventLocation{Little Whinging, United Kingdom}
\EventLogo{}
\SeriesVolume{42}
\ArticleNo{23}

\begin{document}

\maketitle

\begin{abstract}
 We study a class of reachability problems in weighted graphs with
 constraints on the accumulated weight of paths.  The problems we
 study can equivalently be formulated in the model of vector addition
 systems with states (VASS).  We consider a version of the
 vertex-to-vertex reachability problem in which the accumulated weight
 of a path is required always to be non-negative.  This is equivalent
 to the so-called control-state reachability problem (also called  the coverability problem) for 1-dimensional
 VASS. We show that this problem lies in \NC: the class of
 problems solvable in polylogarithmic parallel time.  In our main result we
 generalise the problem to allow disequality constraints on edges
 (i.e., we allow edges to be disabled if the accumulated weight is
 equal to a specific value).  We show that in this case the
 vertex-to-vertex reachability problem is solvable in polynomial time
 even though a shortest path may have exponential length.  In the
 language of VASS this means that control-state reachability is in
 polynomial time for 1-dimensional VASS with disequality tests.
\end{abstract}


\section{Introduction}\label{sec:intro}
In this paper we study reachability problems in weighted graphs with
constraints on the accumulated weight along a path.  We show that the
vertex-to-vertex reachability problem is in \NC{} if the constraint is
that the accumulated weight must always be non-negative, and the problem is
in polynomial time if we additonally allow disequality constraints on
edges (i.e., constraints that prevent an edge from being taken in a path if
the accumulated weight prior to taking the edge is equal to a specific
value).  In both cases a shortest path satisfying the constraints may
have length exponential in the problem description.  Several related
problems have been studied in the literature, including the problem of
finding a path from a source vertex to target vertex that has a
specific total weight~\cite{NykanenU02}.

The problems we study can naturally be formalised as reachability
problems for types of one-counter machines, and the majority of the
related work has been presented in this context.  Under this
correspondence, the value of the counter represents the accumulated
weight along a path, and tests on the counter encode constraints on
allowable paths.  Algorithmic properties of one-counter machines have
been studied by many authors over several
decades~\cite{BolligQS17,BundalaO17,DemriLS10,FearnleyJ13,FinkelGH13,GollerHOW10,HKOW,LafourcadeLT04}.
The above references are a small subset of the extensive literature on
one-counter machines, but they well illustrate that there are many
variations on the basic model and that these variations can lead to
the model having substantially different algorithmic properties.
Particular features mentioned in the references above, driven by
applications to automated verification and program analysis, include
equality tests, disequality tests, inequality tests, parametric tests,
binary updates, polynomial updates, and parametric updates.  

Analysing the complexity of reachability in the presence of the
features listed above leads to a rich complexity landscape.
It is shown in~\cite{LafourcadeLT04} that control-state reachability
is decidable in \NL{} for a ``plain vanilla'' model of one-counters
machine---namely with a counter taking values in the nonnegative
integers with operations increment, decrement, and zero
testing.  Thinking of one-counter machines as one-dimensional vector
addition systems with states (1-VASS), it is natural to allow the
counter to be updated by adding integer constants in binary.  In this
case, still with equality tests, control-state reachability becomes
\NP-complete~\cite{HKOW}.  The \NP{} upper bound here is non-trivial
since, due to the binary encoding of integers, a computation that
reaches the goal state may have length exponential in the size of the
machine.  If one enriches the model further by introducing
inequality tests (comparing the counter with an integer constant) then
control-state reachability becomes \PSPACE-complete~\cite{FearnleyJ13}.
A model of intermediate complexity is one with equality and
disequality tests (introduced in~\cite{DemriLS10}, with applications
to temporal-logic model checking).  In this case the complexity of
control-state reachability is open (between \NP{} and \PSPACE).

In this paper we consider 1-VASS with disequality tests, but no
equality tests.  In terms of 1-VASS, our main result states that the
control-state reachability problem is solvable in polynomial time for
1-VASS with disequality tests.  This result confirms the intuition
that disequality tests are weaker than equality tests.  The main
technical challenge to obtaining a polynomial-time bound is that a run
witnessing that a given control state is reachable may have length
exponential in the description of the counter machine.  A standard
way to overcome this obstacle in related settings is to show that one
may restrict attention to computations that fit a regular pattern
(usually in terms of iterating a ``small'' number of cycles).  Here
the presence of disequality tests proves to be surprisingly
disruptive: it destroys the monotonicity of the transition relation
and prevents from freely iterating positive-weight cycles.  (For
example, the lack of monotonicity means that it is \coNP{} hard to
determine 
whether, given a control state~$s_0$,   for all counter values~$u\in \mathbb{N}$ the configuration~$(s_0,u)$ is unbounded, i.e., can reach infinitely many configurations---see Figure~\ref{fig:GCD}---whereas the same problem for
1-VASS without tests is easily seen to be decidable in polynomial
time.)  Resolving the complexity of reachability for 1-VASS with both
equality and disequality tests remains open.  We hope that the
techniques developed here can help solve this challenging problem.

To complement our main result, we show that for 1-VASS without tests
control-state reachability (and hence also boundedness) is decidable
in \NC{}, i.e., the subclass of \P{} consisting of problems solvable in
polylogarithmic parallel time.  Problems in \NC{} are in particular
solvable in polylogarithmic space.  Related to this, Rosier and
Yen~\cite{RosierY86} have shown that boundedness for VASS is
\NL-complete in case there are absolute bounds on the dimension and
bit-size of integer vectors.  


 \begin{figure}[t] \begin{center} \scalebox{.85}{\begin{tikzpicture}[->,>=stealth',shorten >=.7pt,node distance=1cm, initial text={},every text node part/.style={align=center}]

\node [state,rectangle,initial right](0,0) (s0)  {\scriptsize{$s_0$}};
\node [state,rectangle](0,0) (s1) [below left=0.4cm and 3.5cm of s0]  
{\scriptsize{$s_1$}};
\node [state,rectangle](0,0) (s2) [below left=0.5cm and 0.4cm of s0]  
{\scriptsize{$s_2$}};
\node [](0,0) (dots) [below right=0.6cm and -0.5cm of s0]  
{\huge{$\cdots$}};
\node [state,rectangle](0,0) (sm) [below right=0.4cm and 2.5cm of s0]  
{\scriptsize{$s_n$}};

\path	 (s0) edge node [midway,above]  {} (s1);
\path	 (s0) edge node [midway,above]  {} (s2);
\path	 (s0) edge node [midway,above]  {} (sm);

\path (s1) edge [out=155, in= 200, looseness=6] node[midway,left]{$c_1$} (s1);
\path (s2) edge [out=155, in= 200, looseness=6] node[midway,left]{$c_2$} (s2);
\path (sm) edge [out=-25, in= 20, looseness=6] node[midway,right]{$c_m$} (sm);

%
\end{tikzpicture}} \end{center} 
\vspace{-.5cm} \caption{A
 1-VASS with disequality tests, derived from a 3-CNF formula
 $\varphi$ having propositional variables $X_1,\ldots,X_m$ and clauses
 $C_1,\ldots,C_n$.   We have states $s_1,\ldots,s_n$---one
 state for each clause---and an initial state $s_0$.
The reduction is such that $(s_0,u)$ is
 unbounded for all $u\in\mathbb{N}$ iff $\varphi$ is
 unsatisfiable.
 Let $p_1,\ldots,p_m$ be the first $m$ primes and
 write $P:=p_1\cdots p_m$ for their product. 
 For all~$u\in \mathbb{N}$, define the propositional assignment~$\mathrm{val}_u: \{
 X_1,\ldots,X_m\} \rightarrow \{0,1\}$ by $\mathrm{val}_u(X_i)=1$ if
 and only if $p_i \mid u$.
    Suppose that
 state $s$ corresponds to a clause $C$ that mentions variables
 $X_{i_1},X_{i_2},X_{i_3}$.  Then we place a self-loop on $s$ with
 increment $c_i:=p_{i_1}p_{i_2}p_{i_3}$ and add disequlity tests on~$s$ (or equivalently on the self-loop on~$s$) for all those values  $u \in \{P,P+1,\ldots,P+p_{i_1}p_{i_2}p_{i_3}-1\}$
 where the assignment $\mathrm{val}_u$
  satisfies the clause~$C$. Given
 $u\in \{0,1,\ldots,P-1\}$,  observe that the configuration~$(s_0,u)$ is bounded iff
 $\mathrm{val}_u$ satisfies $\varphi$ (see Appendix~\ref{apx:gcd} for a complete proof).}\label{fig:GCD} \end{figure}
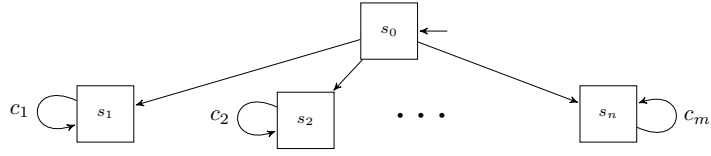

\section{Definitions}\label{sec:def}
We write $\mathbb{N}$ to denote the set of all nonnegative integers
$0,1,2,\dots$ In presenting our results we assume familiarity of the reader with
basic graph theory and computational complexity. 

\subparagraph*{One-Dimensional Vector Addition Systems with States and Tests.}
A \emph{1-VASS with disequality tests} is a tuple $\VASS=(Q,D,\Delta,w)$, where $Q$ is a set of
\emph{states}, $D=\{D_q\}_{q\in Q}$ is a 
collection of cofinite subsets $D_q\subseteq \mathbb{N}$,
$\Delta\subseteq Q\times Q$ is a set of \emph{transitions}, and
$w:\Delta \rightarrow \mathbb{Z}$ is a function that assigns an integer
weight to each transition.
In the special case that each $D_q$ equals  $\mathbb{N}$, we simply call~$\VASS$
a 1-VASS (and we omit the collection $D$).

A \emph{configuration} of $\VASS$ is a pair $(q,z)$ comprising a state~$q\in Q$
and a nonnegative integer~$z\in \mathbb{N}$ referred to as the
\emph{counter value}.
We write $\Conf$ for the set $Q\times \mathbb{N}$ of all configurations.  We
define a partial order on $\Conf$ by $(q,z) \leq (q',z')$ if and only if $q=q'$
and $z\leq z'$.  A configuration $(q,z)$ is \emph{valid} if $z\in D_q$.

A \emph{path} in $\VASS$ is a sequence of states
$\pi=q_1,\ldots,q_n$ such that $(q_i,q_{i+1}) \in \Delta$ for all $i\in\{1,\ldots,n-1\}$.
We sometimes refer to such a path as a $q_1$-$q_n$ path.
Let  $\pi' = p_1,p_2,\ldots,p_m$ be another path
such that $q_n=p_1$, we define
$\pi_1 \cdot \pi_2 := q_1,\ldots,q_n,p_2,\ldots,p_m$.  
Given states $p,q,r$, a set $P$ of $p$-$q$ paths, and a set $R$ of $q$-$r$ paths,
we define $P\cdot R:=\{ \pi \cdot \pi' \mid \pi \in P,\pi' \in R\}$.
The weight of $\pi$ is defined to be $\mathrm{weight}(\pi)
:= \sum_{i=1}^{n-1} w(q_i,q_{i+1})$.  A (possibly empty) prefix of~$\pi$ is said to be \emph{minimal} if it has minimal weight among all
prefixes of $\pi$.  Define $\mathrm{pmin}(\pi)$ to be the
weight of a minimal prefix of $\pi$.

A \emph{run} is a sequence $(q_1,z_1),\ldots,(q_n,z_n)$ of
configurations of $\VASS$ such that there is a path
$\pi=q_1,\ldots,q_n$ with $z_{i+1}=z_i + w(q_i,q_{i+1})$ for
$i=1,\ldots,n-1$.  
We write $(q_1,z_1) \runto{\pi} (q_n,z_n)$ to denote such a run. 
Observe that runs are not allowed to reach negative counter values.
A \emph{valid run} is a run whose configurations are all
valid.  Intuitively, a valid run through $q$ can proceed if and only
if the current counter value is in $D_q$.
We say that a configuration $(q',z')$ is \emph{reachable} from $(q,z)$ if there is a valid run $\pi$ such that $(q,z) \runto{\pi} (q',z')$.

In computational problems  all numbers in the
description of~$\VASS$ are given in binary.
Given a state~$q$ we represent the cofinite set~$D_q$ as the complement of
an explicitly given subset of~$\mathbb{N}$. 
Given this convention, we can assume without loss of generality  that for all states~$q$ the set~$D_q$ is
either~$\mathbb{N}$ or $\mathbb{N}\setminus \{g\}$ for some $g\in\mathbb{N}$; see Appendix~\ref{apx:guard}.
For states~$q$ with $D_q=\mathbb{N}\setminus \{g\}$, we refer to 
the single missing value~$g$ in the domain as the \emph{disequality guard} on~$q$.


\subparagraph*{The Coverability and Unboundedness Problems.}
Let $\VASS=(Q,\Delta,D,w)$ be a 1-VASS with disequality tests, and let $s$ and $t$ be
two  distinguished states  of~$\VASS$.
The \emph{Coverability Problem} asks whether
there exists a valid run in $\VASS$ from $(s,0)$ to $(t,z)$ for some~$z\in
\mathbb{N}$
(in which case we say that $(s,0)$ can \emph{cover} $t$).
The \emph{Unboundedness Problem} asks whether the set of configurations reachable from $(s,0)$ is
infinite (in which case we say that $(s,0)$ is \emph{unbounded}).

The Coverability problem reduces to the Unboundedness problem by, intuitively, 
forcing~$(t,0)$ to be unbounded using a positive cycle, and removing all states that 
cannot reach~$t$ in the underlying graph of $\VASS$. 
In fact, the following holds.

\begin{lemma}\label{lem:cov2unbound}
  There is an $\NC^2$-computable many-one reduction from the
  Coverability Problem to the Unboundedness Problem.
\end{lemma}

Henceforth, we focus on the complexity of deciding the Unboundedness Problem. In
Section~\ref{sec:1vassdiseq} we prove that the Unboundedness Problem for 1-VASS
with disequality tests is decidable in polynomial time. Since $\NC^2
\subseteq \P$, by Lemma~\ref{lem:cov2unbound} we also have that
the Coverability Problem in this setting
is decidable in polynomial time.
In Section~\ref{sec:1vass} we prove that the Unboundedness Problem for 1-VASS
(without disequality tests) is in $\NC^2$, and 
we deduce  that the Coverability Problem for 1-VASS is decidable in $\NC^2$.

\section{Unboundedness for 1-VASS with Disequality Tests}
\label{sec:1vassdiseq}
Fix a $1$-VASS $\VASS = (Q,D,\Delta,w)$ with disequality tests and a
distinguished state $s \in Q$. We are interested in determining whether the
configuration $(s,0)$ is unbounded.

For a (possibly infinite) path $\pi=q_1,q_2,\ldots$, denote by $\blocked(\pi)$
the set of $z\in\mathbb{N}$ such that the unique induced run from $(q,z)$ either contains a
negative counter value or violates a disequality guard. That is, $\pi$ does
not lift to a valid run from the configuration $(q_1,z)$.

\begin{example}
In Figure~\ref{fig:example}, 
since~$41$ is the  guard on~$s_5$ the run $(s_4,93),(s_5,41),(s_6,93)$
is not valid and $93 \in \blocked(s_4,s_5,s_6)$. 
Observe that $\blocked(s_4,s_5,s_6)=[0,52) \cup \{90,93,96\}$ and
$\blocked((s_4,s_5,s_6)^{\omega})=
[0,52) \cup \{52\leq z\leq 96 \mid z \equiv 0,3,6 \pmod{9}\}$.
\end{example}

Recall that for a path~$\pi$, $\pmin(\pi)$ is the weight of a
minimum-weight prefix of $\pi$. 
Let $Q_+ \subseteq Q$ be the set of
states $q \in Q$ such that there is a positive-weight simple cycle on
$q$ in the underlying graph of $\VASS$.
  For $q \in Q_+$ we pick a
simple cycle $\gamma_q$ such that $\pmin(\gamma_q) \geq \pmin(\gamma)$
for any other positive-weight simple cycle~$\gamma$ on~$q$; write
$W_q$ for $\weight(\gamma_q)$.\footnote{Note that $\gamma_q$ does not necessarily have maximal weight $W_q$ among the positive simple cycles on $q$.}
Define $\Conf_+:=\{
(q,z) \in \Conf \mid q\in Q_{+}, \,{z +\pmin(\gamma_q)} \geq 0\}$.

\begin{figure}[t]
\begin{center}
 \scalebox{.85}{\begin{tikzpicture}[->,>=stealth',shorten >=.7pt,node distance=1cm, initial text={},every text node part/.style={align=center}]

\node [state,rectangle,initial](0,0) (s0)  {\scriptsize{$s_0$}};
\node [state,rectangle](0,0) (s1) [right=1cm of s0]  
{\scriptsize{$s_1$}\\ \scriptsize{$\neq 60$}};
\node [state,rectangle](s2) [above=.6cm of s1] {\scriptsize{$\!s_2\!$}};
\node [state,rectangle](s3) [right=1cm of s1] {\scriptsize{$s_3$}\\ \scriptsize{$\neq 30$}};
\node [state,rectangle](s4) [right=1cm of s3] {\scriptsize{$s_4$}\\ \scriptsize{$\neq 90$}};
\node [state,rectangle](s5) [above left=.6cm and .3cm of s4] {\scriptsize{$s_5$}\\ \scriptsize{$\neq 41$}};
\node [state,rectangle](s6) [above right=.6cm and .3cm of s4] {\scriptsize{$s_6$}\\ \scriptsize{$\neq 96$}};
\node [state,rectangle](s7) [right=1cm of s4] {\scriptsize{$s_7$}\\ \scriptsize{$\neq 70$}};
\node [state,rectangle](s8) [right=1cm of s7] {\scriptsize{$s_8$}\\ \scriptsize{$\neq 80$}};
\node [state,rectangle](s9) [right=1cm of s8] {\scriptsize{$s_9$}\\ \scriptsize{$\neq 80$}};
\node [state,rectangle](s10) [right=1cm of s9] {\scriptsize{$s_{10}$}\\ \scriptsize{$\neq 120$}};
\node [state,rectangle](s11) [above left=.6cm and .35cm of s10] {\scriptsize{$s_{11}$}\\ \scriptsize{$\neq 43$}};
\node [state,rectangle](s13) [above right=.6cm and .35cm of s10] {\scriptsize{$s_{13}$}\\ \scriptsize{$\neq 130$}};
\node [state,rectangle](s12) [above=1.5cm of s10] {\scriptsize{$s_{12}$}\\ \scriptsize{$\neq 130$}};

\path	 (s0) edgenode [midway,below]  {\scriptsize{$12$}} (s1);
\path	 (s1) edge [bend left] node [midway,left]  {\scriptsize{$-12$}} (s2);
\path	 (s2) edge [bend left] node [midway,right]  {\scriptsize{$18$}} (s1);
\path	 (s1) edge node [midway,below]  {\scriptsize{$12$}} (s3);
\path	 (s3) edge node [midway,below]  {\scriptsize{$30$}} (s4);
\path	 (s4) edge [bend left] node [near  end,right]  {\scriptsize{$-52$}} (s5);
\path	 (s5) edge [bend left] node [midway,above]  {\scriptsize{$52$}} (s6);
\path	 (s6) edge [bend left] node [near start,left]  {\scriptsize{$9$}} (s4);
\path	 (s4) edge node [midway,below]  {\scriptsize{$4$}} (s7);
\path	 (s7) edge node [midway,below]  {\scriptsize{$4$}} (s8);
\path	 (s8) edge node [midway,below]  {\scriptsize{$-3$}} (s9);
\path	 (s9) edge node [midway,below]  {\scriptsize{$17$}} (s10);
\path	 (s10) edge [bend left] node [near  end,right]  {\scriptsize{$-80$}} (s11);
\path	 (s11) edge [bend left] node [right,above]  {\scriptsize{$81$}} (s12);
\path	 (s12) edge [bend left] node [left,above]  {\scriptsize{$3$}} (s13);
\path	 (s13) edge [bend left] node [near start,left]  {\scriptsize{$6$}} (s10);

\end{tikzpicture}}
\end{center}
\vspace{-.4cm}
\caption{A 1-VASS with disequality tests. 
Disequality guards are denoted by $\neq$.
For example, in state~$s_1$ the set $D_{s_1}$ is 
$\mathbb{N} \setminus \{60\}$, and no run
goes through~$s_1$ if its current counter 
value is~$60$.}\label{fig:example} 
\end{figure}
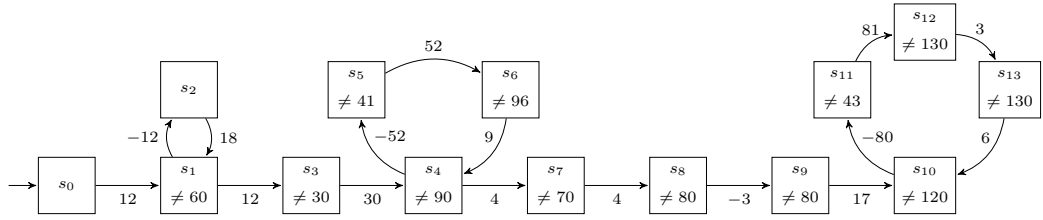

Define a path to be \emph{primitive} if no proper infix is a positive
cycle (note though that a primitive path may itself be a positive
cycle). 
 We say that a run is primitive if the underlying path is
primitive.  Observe that if $\rho$ is a valid run, none of whose internal
configurations (i.e. excluding the first and last configurations) lies in $\Conf_+$, then $\rho$ is primitive.

\begin{example}
 In Figure~\ref{fig:example}, for $s_1 \in Q_+$ we pick the simple cycle $\gamma_{s_1}=s_1,s_2,s_1$ with $W_{s_1}=6$.
  Since $\pmin(\gamma_{s_1})=-12$, we have that $\{z\mid (s_1,z) \in \Conf_+ \} =
  [12,\infty)$.
  Moreover, the path 
  $s_4,s_5,s_6,s_4$ is primitive, but $s_1,s_2,s_1,s_3$ is not primitive.
\end{example}

\begin{proposition}
	\label{prop: unbounded iff confplus unbounded}
A configuration~$(s,0)$ is unbounded if, and only if,
 $(s,0)$ can reach an unbounded configuration in  $\Conf_+$. 	
\end{proposition}

In order to decide whether~$(s,0)$ is unbounded,
by Proposition~\ref{prop: unbounded iff confplus unbounded},
it suffices to compute the set of unbounded configurations in $\Conf_+$ 
and determine whether $(s,0)$ can reach this set.
Define $\Conf_\infty \subseteq \Conf_+$ to be the set of all
unbounded configurations in $\Conf_+$.
Observe that every configuration~$(q,z)\in\Conf_+$ with 
$z\notin \blocked(\gamma_q^\omega)$ can take the cycle
$\gamma_q$  arbitrarily many times and is thus included in $\Conf_\infty$.
However, even if $z\in \blocked(\gamma_q^\omega)$, it may still be the case that
$(q,z)$ is unbounded, by traversing more complicated paths.

\begin{example}
  In Figure~\ref{fig:example}, all  configurations~$(s_4,z)$ with $z$ in
  $\mathbb{N} \setminus \blocked((s_4,s_5,s_6)^{\omega})= \{52\leq z\leq 96 \mid
  z \not \equiv 0,3,6 \pmod{9}\} \cup (96,\infty)$ are trivially unbounded and
  thus included in~$\Conf_{\infty}$.  
  It will transpire that
  $\{s_4\}\times \{54,60,63,69\} \subseteq \Conf_{\infty}$ even though
  $\{54,60,63,69\} \in \blocked((s_4,s_5,s_6)^{\omega})$.
\end{example}

In order to reason about the aforementioned complicated paths, we proceed as follows. In Section~\ref{subsec: classes and chains} we introduce residue classes and chains, which form a partition of $\Conf_+$, and are the building blocks of our analysis. In Section~\ref{subsec: inductive construction} we characterize $\Conf_\infty$ as the limit of an inductive construction. This  enables us to reason about the structure of $\Conf_\infty$ in Section~\ref{subsec:struct}. Finally, in Section~\ref{subsec:algo} we show how to compute $\Conf_\infty$ and  decide unboundedness.

\subsection{Residue Classes and Chains}
\label{subsec: classes and chains}
Given $q \in Q_+$ and $0
\leq r < W_q$, we call the set of configurations $\{ (q,z) \in \Conf_+ \mid
z\equiv r \pmod{W_q} \}$ a \emph{$q$-residue class}.  We simply speak
of a \emph{residue class} if we do not want to specify the state $q$.
Given a $q$-residue class $R$, a set $C\subseteq R$ is called
a \emph{$q$-chain} if it is a maximal subset of $R$ for the property
that every pair of configurations $(q,z),(q,z') \in C$ with $z<z'$ is
connected by a valid run obtained by iterating the cycle $\gamma_q$. Again,
we speak of a \emph{chain} if we do not want to specify the state $q$.

\begin{figure}[t]
\begin{center}
 \scalebox{.85}{\begin{tikzpicture}[->,>=stealth',shorten >=.7pt,node distance=1cm, initial text={},every text node part/.style={align=center},yscale=.8]

\node [state,rectangle](0,0) (q1) 
{\scriptsize{$s_1 ~~(\neq 60)$}};
\node [state,rectangle] (q4) [right=2.2cm of q1] 
{\scriptsize{$s_4~~(\neq 90,93,96)$}};
\node [state,rectangle] (q10) [right=2.2cm of q4] 
{\scriptsize{$s_{10}~~(\neq 120,123,126,129)$}};
\path	 (q1) edge [loop above] node {6} (q1);
\path	 (q4) edge [loop above] node {9} (q4);
\path	 (q10) edge [loop above] node {10} (q10);
%


\draw[dotted, thick,blue!90!gray] (.4,-.8)-- ++ (0,.3);
\draw[-,dotted, thick,magenta!90!gray] (.4,-9.9)-- ++ (0,8.2);

\foreach \i in {0,...,4}{
 \draw[dotted, thick,blue!90!gray] (-.35+.15*\i,-9.9)-- ++ (0,9.4);
}
\fill [blue!60!white] (.4,-.8) circle (.8mm) node {\textcolor{blue!60!white}{\scriptsize{$ \quad \quad 66$}}};
\draw[-,red, thick] (-.5,-10)node[left]{\scriptsize{$12\!$}}-- ++ (1.2,0) ;

\foreach \i/\j in {0/12,1/18,2/24,3/30,4/36,5/42,6/48,7/54,8/60 } {
      \fill [blue!60!white] (-.35,\i-9) circle (.8mm);
      \fill [blue!60!white] (-.2,\i-9.15) circle (.8mm);
      \fill [blue!60!white] (-.05,\i-9.3) circle (.8mm);
		\fill [blue!60!white] (.1,\i-9.45) circle (.8mm);
      \fill [blue!60!white] (.25,\i-9.6) circle (.8mm);
      \draw [magenta,fill=magenta!10!white] (.4,\i-9.75) 
		circle (.8mm) node {\textcolor{magenta}{\scriptsize{$ \quad \quad$\j}}};
}
	

\foreach \i in {0,1,3,4,6,7}{
\draw[dotted, thick,blue!90!gray] (3.5+.15*\i,-8.3)-- ++ (0,7.8);
}
\foreach \i in {2,5,8}{
\draw[-,dotted, thick,magenta] (3.5+.15*\i,-8.3)-- ++ (0,6.35-.15*\i);
}
\draw[dotted, thick,blue!90!gray] (3.8,-.8)-- ++ (0,.3);
\fill [blue!60!white] (3.8,-.8) circle (.8mm);
\draw[dotted, thick,blue!90!gray] (4.25,-1.25)-- ++ (0,.75);
\fill [blue!60!white] (4.25,-1.25) circle (.8mm);
\draw[dotted, thick,blue!90!gray] (4.7,-1.7)-- ++ (0,1.2);
\fill [blue!60!white] (4.7,-1.7) circle (.8mm)
node {\textcolor{blue!60!white}{\scriptsize{$ \quad \quad 99$}}};
\fill [blue!60!white] (3.5,-8) circle (.8mm);
\fill [blue!60!white] (3.65,-8.15) circle (.8mm);
\draw[-,red, thick] (3.3,-8.4) node[left]{\scriptsize{$52\!$}}-- ++ (1.65,0) ;
%
\foreach \i/\j in {0/54,1/63,2/72,3/81,4/90}{

      \fill [blue!60!white] (3.5,1.2*\i-6.65) circle (.8mm);
      \fill [blue!60!white] (3.65,1.2*\i-6.8) circle (.8mm);
		 \draw [magenta,fill=magenta!10!white] (3.8,1.2*\i-6.95) 
		circle (.8mm);
			\fill [blue!60!white] (3.95,1.2*\i-7.1) circle (.8mm);
      \fill [blue!60!white] (4.1,1.2*\i-7.25) circle (.8mm);
		\draw [magenta,fill=magenta!10!white] (4.25,1.2*\i-7.4) 
		circle (.8mm);
	   \fill [blue!60!white] (4.4,1.2*\i-7.55) circle (.8mm);
      \fill [blue!60!white] (4.55,1.2*\i-7.7) circle (.8mm);
      \draw [magenta,fill=magenta!10!white] (4.7,1.2*\i-7.85) 
		circle (.8mm) node {\textcolor{magenta}{\scriptsize{$ \quad \quad$\j}}};
}

\draw[dotted, thick,blue!90!gray] (11.5-3,-.8)-- ++ (0,.3);
\fill [blue!60!white] (11.5-3,-.8) circle (.8mm);
\draw[dotted, thick,blue!90!gray] (11.95-3,-1.25)-- ++ (0,.75);
\fill [blue!60!white] (11.95-3,-1.25) circle (.8mm);
\draw[dotted, thick,blue!90!gray] (12.4-3,-1.7)-- ++ (0,1.2);
\fill [blue!60!white] (12.4-3,-1.7) circle (.8mm);
\draw[dotted, thick,blue!90!gray] (12.85-3,-2.15)-- ++ (0,1.65);
\fill [blue!60!white] (12.85-3,-2.15) circle (.8mm)
node {\textcolor{blue!60!white}{\scriptsize{$ \quad \quad 130$}}};

\draw[-,red, thick] (11.3-3,-9.4) node[left]{\scriptsize{$80\!$}}-- ++ (1.8,0);

\foreach \i in {0,1,3,4,6,7}{
\draw[dotted, thick,blue!90!gray] (11.65-3+.15*\i,-9.3)-- ++ (0,8.8);
}
\foreach \i in {0,3,6,9}{
\draw[-,dotted, thick,magenta] (11.5-3+.15*\i,-9.3)-- ++ (0,7-.15*\i);
}
\foreach \i/\j in {-1/80,0/90,1/100,2/110,3/120}{
		\draw [magenta,fill=magenta!10!white] (11.5-3,1.35*\i-6.4) circle (.8mm);
      \fill [blue!60!white] (11.65-3,1.35*\i-6.55) circle (.8mm);
      \fill [blue!60!white] (11.8-3,1.35*\i-6.7) circle (.8mm);
      \draw [magenta,fill=magenta!10!white] (11.95-3,1.35*\i-6.85) circle (.8mm);
		\fill [blue!60!white] (12.1-3,1.35*\i-7) circle (.8mm);
      \fill [blue!60!white] (12.25-3,1.35*\i-7.15) circle (.8mm);
      \draw [magenta,fill=magenta!10!white] (12.4-3,1.35*\i-7.3) circle (.8mm);
		\fill [blue!60!white] (12.55-3,1.35*\i-7.45) circle (.8mm);
      \fill [blue!60!white] (12.7-3,1.35*\i-7.6) circle (.8mm);
      \draw [magenta,fill=magenta!10!white] (12.85-3,1.35*\i-7.75) 
    circle (.8mm) node {\textcolor{magenta}{\scriptsize{$ \quad \quad$\j}}};
}

\end{tikzpicture}}
\end{center}
\vspace{-.5cm}
\caption{We focus on states~$s_1$, $s_4$, and $s_{10}$ in the 1-VASS in Figure~\ref{fig:example},
each of which lies on a  simple positive cycle.
We also indicate which counter values
prevent taking the associated positive cycle.
For example, state~$s_4$
has the simple  cycle~$\gamma_{s_4}$ with  $W_{s_4}=9$ and
taking $\gamma_{s_4}$ from~$\{s_4\} \times \{90,93,96\}$ is not allowed 
due to disequality guards along~$\gamma_{s_4}$.
The columns underneath each 
state represent residue classes of that state in~$\Conf_{+}$.
We colour all unbounded chains in blue and all bounded chains in pink;
thus all blue configurations form the set~$U_0$.  
}\label{fig:U0} 
\end{figure}
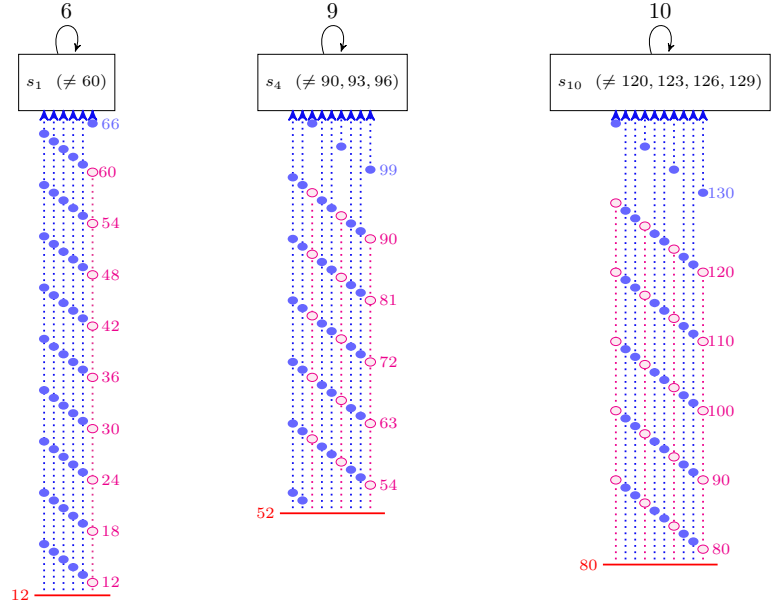

We draw a distinction between \emph{bounded chains} and \emph{unbounded chains},
where a chain is bounded if and only if the associated set of counter values is
bounded.  An unbounded $q$-chain $C$ is contained in $\Conf_\infty$ since the
cycle $\gamma_q$ can be taken arbitrarily many times from any configuration in
$C$ to yield a valid run.

\begin{remark}\label{lem:2q}
  Let us write $\gamma_q = q_1, q_2, \dots$
  For each $q_1$-residue class $R$, every $z$ such that 
  $z + \mathrm{weight}(q_1,\dots,q_i) \not\in D_{q_i}$, for some $q_i$,
  induces at most two bounded chains. Namely, the set of configurations
  below $(q,z)$ form a chain; and the singleton $\set{(q,z)}$ is also
  (vacuously) a chain. Note that every residue class also has one unbounded
  chain. That is, the set of configurations above~$(q,z)$ with $z$ the maximal
  ``induced guard'' on~$q$. Since there are at most $|Q|$ guards, each residue
  class decomposes as a disjoint union of at most~$2\abs{Q}$ bounded chains
  and a single unbounded chain.
\end{remark}
Intuitively, within each bounded chain we can iterate the cycle $\gamma_q$ 
until hitting a guard.
We call a residue class
$R$ \emph{trivial} if it consists solely of a single unbounded chain.
Note that the union of all bounded $q$-chains is equal to  $\Conf_{+} \cap \{q\}
\times \blocked(\gamma_q^{\omega})$.

\begin{example}
  As indicated in Figure~\ref{fig:U0} for the running example,  the residue
  classes $\{s_4\}\times (52+i+9\NN)$ with $i\in \{0,1,3,4,6,7\}$ are indeed
  trivial, while each residue class $\{s_4\}\times (52+i+9\NN)$ with $i\in
  \{2,5,8\}$ consists of two bounded chains  $\{s_4\}\times \{52 \leq z < 88+i
  \mid  z \equiv i \pmod{9} \}$ and $\{s_4\}\times \{88+i\}$, and a single
  unbounded chain $\{s_4\}\times (88+i+9\NN)$.
\end{example}

One of the main ideas in this section is to show that a configuration is
unbounded if and only if it can reach an unbounded chain via a valid run whose
underlying path $\pi$ has the form
\[
  \pi = \pi_0 \cdot \gamma_{q_1}^{n_1} \cdot \pi_1 \cdots \pi_{k-1} \cdot
  \gamma_{q_k}^{n_k} \cdot \pi_{k} \, ,
\]
where $\pi_0,\ldots,\pi_{k}$ are primitive paths and $n_1,\ldots,n_k$ are
non-negative integers.  Moreover, we give a polynomial bound on  the length of
the~$\pi_i$ and the magnitude of $k$ in terms of the size of the underlying
1-VASS  (in general, the exponents $n_i$ may be exponential in the size of the
1-VASS).  We also  show how to detect the existence of such a path in
polynomial time.

Recall the structure of $\Conf$ as a partially ordered set. 
We will use standard order-theoretic terminology and notation to refer
to sets of configurations: in particular given  sets of configurations
$S, S' \subseteq \Conf$, we say that $S$ is \emph{downward closed in $S'$}
if for all $(q,z) \in S\cap S'$ and $(q,z') \in S'$ with $z'\leq z$, we have 
$(q,z') \in S$.

\subsection{Inductive Characterization of $\Conf_\infty$}
\label{subsec: inductive construction}
We now give an inductive backward-reachability construction of  the
set of all configurations in $\Conf_+$ that can reach an unbounded
chain. 
Since unbounded configurations can, in particular, 
reach unbounded chains (as above the maximal disequality guard, all chains are unbounded), this set is exactly $\Conf_{\infty}$.

In order for our inductive construction to converge in a polynomial number of
steps, we essentially consider meta-transitions of the form
$\gamma_q^k \cdot \pi$ for $\gamma_q$ a simple cycle,
$k\in \mathbb{N}$, and~$\pi$ a primitive path.  Formally, we define an
increasing sequence $U_0 \subseteq U_1 \subseteq U_2 \subseteq \cdots$
of subsets of $\Conf_+$ such that
$\bigcup_{n\in \mathbb{N}} U_n = \Conf_\infty$.  Define $U_0$ to be
the union of the collection of unbounded chains.  
Given $n\in\mathbb{N}$ we inductively construct $U_{n+1} $ as follows.
First, define $U_n' \subseteq \Conf_+$ 
as the set of
configurations $(q,z) \not\in U_n$ whose distance to $U_n$ is minimal
among all configurations in $\Conf_+ \setminus U_n$ (here the distance
of a configuration $(q,z)$ to $U_n$ is the length of the shortest
valid run from $(q,z)$ to $U_n$). 
Now define $U_{n+1} \subseteq \Conf_+$ to be the smallest set
  such that $U_n,U_n'\subseteq U_{n+1}$ and $U_{n+1}\cap C$ is downward closed
in every chain $C$.
Then $\bigcup_{n\in\mathbb{N}} U_n$ is the set of configurations in $\Conf_+$ that can 
reach an unbounded chain which, as noted above, is equal to $\Conf_\infty$.

\begin{remark}
	\label{rmk: Un primitive to Un+1}
By definition, a shortest run from a configuration
          $(q,z)\in U'_{n+1}\setminus U_n$ to $U_n$ has no internal configurations in
          $\Conf_+$, and is therefore
          primitive. 
\end{remark}

\begin{figure}[t]
\begin{subfigure}[t]{0.64\textwidth}
  \centering
  \scalebox{.8}{\begin{tikzpicture}[->,>=stealth',shorten >=.7pt,node distance=1cm, initial text={},every text node part/.style={align=center},yscale=.85]


\draw[->,green!70!black](3.8,-4.4+3-8.95)--  (11.5,-3.2+1.35*1-8.45) ;
\draw[->,green!70!black](4.25,-4.4+3-9.4)-- node[near end,right]{\scriptsize{$22$}} (11.95,-3.2+1.35*1-8.9);
\draw[->,green!70!black](4.7,-4.4+3-9.85)-- node[near end,right]{\scriptsize{$22$}} (12.4,-3.25+1.35*1-9.35);
\draw[->,green!70!black](3.8,-4.4+2-8.95)-| node[near start,above]{\scriptsize{$22$}}  (12.85,-3.1+1.35*1-9.8);
\draw[->,green!70!black](4.25,-4.4+2-9.4) --  
node[near end,above]{\scriptsize{$4+4-3$}} (8.25,-4.4+2-9.4) node[right]{\scriptsize{blocked at $(s_9,80)$}};
\draw[->,green!70!black](4.7,-4.4+2-9.85) --  
node[near end,above]{\scriptsize{$4+4$}} (8.25,-4.4+2-9.85) node[right]{\scriptsize{blocked at $(s_8,80)$}};
\draw[->,green!70!black](3.8,-4.4+1-8.95)-| (12.7,-4.2-8.65)  ;
\draw[->,green!70!black](4.25,-4.4+1-9.4) --  
node[near end,above]{\scriptsize{$4$}} (8.25,-4.4+1-9.4) node[right]{\scriptsize{blocked at $(s_7,70)$}};

\draw[->,green!70!black](4.7,-4.4+1-9.85)-| (12.1,-3.2-1.35-9.05)  ;

\foreach \i in {2,5,8}{
\draw[-,dotted, thick,magenta] (3.5+.15*\i,-3.2-11.1)-- ++ (0,4.35-.15*\i);
}
\draw[-,red, thick] (3.3,-14.5) node[left]{\scriptsize{$52\!$}}-- ++ (1.65,0) ;
\node[draw=none,align=center] at (4,-15.2){bounded $s_4$-chains\\(excluding guards)};

\foreach \i/\j in {0/54,1/63,2/72,3/81}{
 
		 \draw [magenta,fill=magenta!10!white] (3.8,-4.4+\i-8.95) 
		circle (.8mm);
		\draw [magenta,fill=magenta!10!white] (4.25,-4.4+\i-9.4) 
		circle (.8mm);
      \draw [magenta,fill=magenta!10!white] (4.7,-4.4+\i-9.85) 
		circle (.8mm) node {\textcolor{magenta}{\scriptsize{$ \quad \quad$\j}}};
}	

\foreach \i in {0,1}{
	\draw [green!70!black,fill=green!70!black] (3.8,-4.4+\i-8.95) 
		circle (.8mm);  
			\draw [green!70!black,fill=green!70!black] (4.7,-4.4+\i-9.85) 
		circle (.8mm);
}

\draw[-,red, thick] (11.3,-14.5) node[left]{\scriptsize{$80\!$}}-- ++ (1.8,0);
\node[draw=none,align=center] at (12,-15.2){relevant section\\of  $s_{10}$-chains };

\foreach \i in {0,3,6,9}{
\draw[-,dotted, thick,magenta] (11.5+.15*\i,-3.2-11.2)-- ++ (0,4.2-.15*\i);
}
\foreach \i/\j in {-1/80,0/90,1/100}{
		\draw [magenta,fill=magenta!10!white] (11.5,-3.2+1.35*\i-8.45) circle (.8mm);
      \fill [blue!60!white] (11.65,-3.2+1.35*\i-8.6) circle (.8mm);
      \fill [blue!60!white] (11.8,-3.2+1.35*\i-8.75) circle (.8mm);
      \draw [magenta,fill=magenta!10!white] (11.95,-3.2+1.35*\i-8.9) circle (.8mm);
		\fill [blue!60!white] (12.1,-3.2+1.35*\i-9.05) circle (.8mm);
      \fill [blue!60!white] (12.25,-3.2+1.35*\i-9.2) circle (.8mm);
      \draw [magenta,fill=magenta!10!white] (12.4,-3.2+1.35*\i-9.35) circle (.8mm);
		\fill [blue!60!white] (12.55,-3.2+1.35*\i-9.5) circle (.8mm);
      \fill [blue!60!white] (12.7,-3.2+1.35*\i-9.65) circle (.8mm);
      \draw [magenta,fill=magenta!10!white] (12.85,-3.2+1.35*\i-9.8) 
    circle (.8mm) node {\textcolor{magenta}{\scriptsize{$ \quad \quad$\j}}};
}
 
\end{tikzpicture}}
\caption{The set~$U_1$ is obtained from~$U_0$ in
  Figure~\ref{fig:U0}.   
}\label{fig:U1}
\end{subfigure}%
\quad
\begin{subfigure}[t]{0.35\textwidth}
  \centering
  \scalebox{.8}{\begin{tikzpicture}[->,>=stealth',shorten >=.7pt,node distance=1cm, initial text={},every text node part/.style={align=center},yscale=.85]

\draw[->,green!70!black](.4,7-9.75-12.2)-| node[near start,above]{\scriptsize{$42$}} (3.8,1.2*4-6.95-13.3);
\draw[->,green!70!black] (.4,6-9.75-12.2) -| node[near start,above]{\scriptsize{$42$}}  (4.7,1.2*4-7.85-13.3);
\draw[->,green!70!black] (.4,5-9.75-12.2) -- node[near end,above]{\scriptsize{$42$}}  (4.25,1.2*3-7.4-13.3);
\draw[->,green!70!black] (.4,4-9.75-12.2)-- (3.8,1.2*2-6.95-13.3);
\draw[->,green!70!black] (.4,3-9.75-12.2) -|  (4.7,1.2*2-7.85-13.3);
\draw[->,green!70!black] (.4,2-9.75-12.2) -|(4.25,1.2*1-7.4-13.3);

\draw[->,green!70!black](.4,1-9.75-12.2) --  
node[near end,above]{\scriptsize{$12$}} (2,1-9.75-12.2) node[right]{\scriptsize{$(s_3,30)$}};
\draw[->,green!70!black] (.4,-9.75-12.2) -|  (4.7,-7.85-13.3);

\draw[-,dotted, thick,magenta!90!gray] (.4,-9.9-12.2)-- ++ (0,7.2);
\draw[-,red, thick] (-.5,-10-12.2)node[left]{\scriptsize{$12\!$}}-- ++ (1.2,0) ;
\node[draw=none,align=center] at 
 (0.3,-10-12.8){the bounded\\$s_1$-chain};

\foreach \i/\j in {0/12,1/18,2/24,3/30,4/36,5/42,6/48,7/54} {
      \draw [magenta,fill=magenta!10!white] (.4,\i-9.75-12.2) 
		circle (.8mm) node {\textcolor{magenta}{\scriptsize{$ \quad \quad$\j}}};
}


\foreach \i in {2,5,8}{
\draw[-,dotted, thick,magenta] (3.5+.15*\i,-8.3-13.3)-- ++ (0,6.35-.15*\i);
}
\fill [blue!60!white] (3.5,-8-13.3) circle (.8mm);
\fill [blue!60!white] (3.65,-8.15-13.3) circle (.8mm);
\draw[-,red, thick] (3.3,-8.4-13.3) node[left]{\scriptsize{$52\!$}}-- ++ (1.65,0) ;
\node[draw=none,align=center] at 
  (3.8,-10-12.8){relevant section\\of $s_4$-chains};
%
\foreach \i/\j in {0/54,1/63,2/72,3/81,4/90}{

      \fill [blue!60!white] (3.5,1.2*\i-6.65-13.3) circle (.8mm);
      \fill [blue!60!white] (3.65,1.2*\i-6.8-13.3) circle (.8mm);
		 \draw [magenta,fill=magenta!10!white] (3.8,1.2*\i-6.95-13.3) 
		circle (.8mm);
			\fill [blue!60!white] (3.95,1.2*\i-7.1-13.3) circle (.8mm);
      \fill [blue!60!white] (4.1,1.2*\i-7.25-13.3) circle (.8mm);
		\draw [magenta,fill=magenta!10!white] (4.25,1.2*\i-7.4-13.3) 
		circle (.8mm);
	   \fill [blue!60!white] (4.4,1.2*\i-7.55-13.3) circle (.8mm);
      \fill [blue!60!white] (4.55,1.2*\i-7.7-13.3) circle (.8mm);
      \draw [magenta,fill=magenta!10!white] (4.7,1.2*\i-7.85-13.3) 
		circle (.8mm) node {\textcolor{magenta}{\scriptsize{$ \quad \quad$\j}}};
}	
\foreach \i in {0,1}{
	\draw [green!70!black,fill=green!70!black] (3.8,1.2*\i-6.95-13.3) 
		circle (.8mm);  
			\draw [green!70!black,fill=green!70!black] (4.7,1.2*\i-7.85-13.3) 
		circle (.8mm);
}
 \draw [yellow!90!black,fill=yellow!90!black] (.4,-9.75-12.2) 
		circle (.8mm);
		
\end{tikzpicture}}
  \caption{
  The set~$U_2$.
  }\label{fig:U2} 
\end{subfigure}
\caption{The sets $U_1$ and $U_2$ of the running example. The blue configurations are in~$U_0$; green
  ones are in~$U_{1} \setminus U_0$; yellow one is in~$U_{2} \setminus U_1$. The  pink configurations are
  in~$\Conf_{+}\setminus U_{1}$ and~$\Conf_{+}\setminus U_{2}$, respectively.  While computing~$U_1$, the green configurations~$(s_4,63)$ and
  $(s_4,69)$  take the primitive path $\pi=s_4,s_7,s_8,s_9,s_{10}$ to~$U_0$.
  In all other pink configurations in $s_4$-chains, although enabled, the path~$\pi$
  either hits a guard or ends in 
  $(s_{10},z)\in \Conf_{+}\setminus U_1$.}
\end{figure}
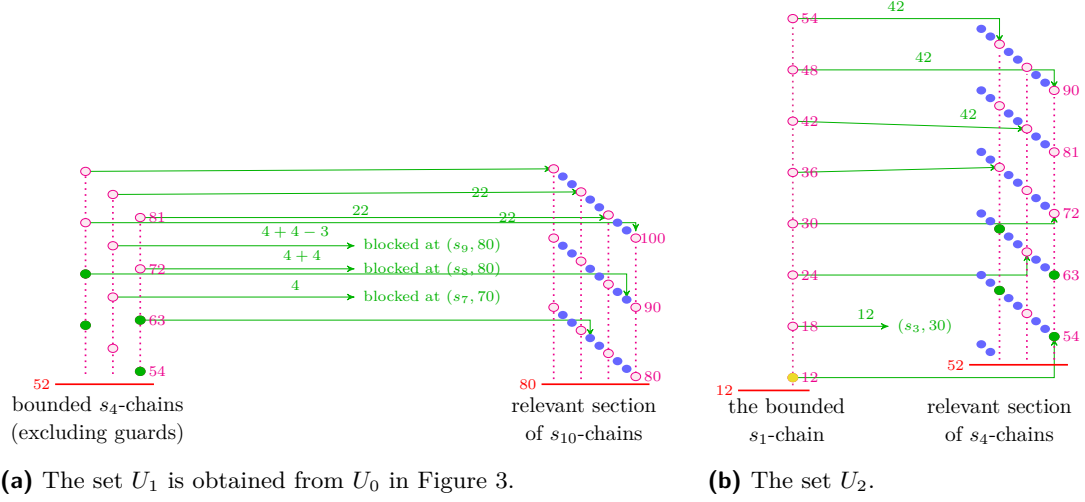

\begin{example}
  Figure~\ref{fig:U0}
  indicates the set~$U_0$ for the running example.
  Note that $U_0$ contains all trivial residue classes.  Observe  that
  $U'_0=\{(s_4,63), (s_4,69)\}$; see Figure~\ref{fig:U1}.  These two
  configurations belong to two distinct  chains. The downward closure
  of~$\{(s_4,63)\}$  in its chain is  $\{s_4\} \times \{54,63\}$, and the
  downward closure of~$\{(s_4,69)\}$  in its chain is  $\{s_4\} \times \{60,69\}$.
  We have that $U_{1} = U_0 \cup (\{s_4\} \times \{54,60,63,69\})$.
  The second iteration to compute~$U_2$
  only adds the configuration~$(s_1,12)$  to $U_1$; see
  Figure~\ref{fig:U2}. 
  The sequence stabilizes in this iteration.
\end{example}

\subsection{The Structure of $\Conf_{\infty}$}\label{subsec:struct}
 In this section we analyze the structure of $\Conf_{\infty}$,
  based on its inductive characterization.  
This analysis will be key in obtaining a polynomial-time algorithm
	to compute $\Conf_\infty$.
	
	
  The guiding intuition is that for all $n$ the set $U_n$ is
  \emph{almost upward closed} in each residue class~$R$.  By this we
  mean that if~$(q,z)$ is the least configuration in $R\cap U_n$, then
  all but polynomially many configurations of~$R$ above~$(q,z)$ are
  also in~$U_n$.  More specifically, we show that for any bounded
  chain $C$ in $R$ that lies above $(q,z)$, although the number of
  configurations in $C$ may be exponential in $\abs{Q}$, the size of
  $C\setminus U_n$ is bounded by a polynomial in $\abs{Q}$.  (Note
  here that the unique unbounded chain in $R$ is contained in $U_0$
  and hence is contained in $U_n$ for all $n\in\NN$.)  Using this
  observation, we provide a polynomial bound on the number of
  iterations until the inductive construction converges. Indeed, in
  every iteration, unless a fixed point has been reached, there must
  exist some bounded chain $C$ such that the size of $C\setminus U_n$
  strictly decreases. After showing that $C\setminus U_n$ is of
  polynomial size,
  we obtain a polynomial bound on the number of
  iterations until $U_n$ converges by Remark~\ref{lem:2q}.

We start by characterizing the paths between chains.
\begin{proposition}
  Let $(q,z),(q',z') \in \Conf_+$ and let $(q,z) \stackrel{\pi}{\rightarrow}
  (q',z')$ be a (not necessarily valid) run such that $\pi$ is a primitive path.
  Then there exists a run $(q,z) \stackrel{\pi'}{\rightarrow} (q',z'')$ of
  length at most $\abs{Q}^2+2$ such that
  \begin{enumerate}
    \item $\pmin(\pi')\ge \pmin(\pi)$,
    \item $z'' \geq z'$, and
    \item the $q'$-residue class of $(q',z'')$ is either trivial or identical to
      that of $(q',z')$.
  \end{enumerate}
  \label{prop:primitive}
\end{proposition}


 Given a $q$-residue class $R$, in general $U_n$ is not an upward
  closed subset of $R$.  The following definitions are intended to
  measure the defect of $U_n$ in this regard.

We say that a bounded chain $C$ that is contained in a residue
  class $R$ is \emph{$n$-active} if there exists a configuration in
  $U_n \cap R$ that lies below some configuration in $C$.
 Let $C$ be an $n$-active chain.  Recall that
  $U_n$ is downward closed in $C$ and hence $C\setminus U_n$ is upward
  closed in $C$.  Suppose that $C\setminus U_n$ is non-empty, write
  $m_1 := \min \{ x : (q,x) \in C\setminus U_n \}$ and
  $m_2:=\max \{ x: (q,x) \in C\setminus U_n\}$, and define\footnote{We omit $q$ from the definition of $\delta_n(C)$ for brevity.}
 \[
   \delta_n(C):= \{ (q,x) \in \Conf_+ : m_1\leq x\leq m_2 \text{ and } (q,x)
\not\in U_n\}.
 \]
Thus $\delta_n(C)$ contains all configurations in
  $C\setminus U_n$, as well as all configurations ``between''
  elements of $C\setminus U_n$, apart from those that are themselves
  in $U_n$.  If $C\setminus U_n=\emptyset$ then we define
  $\delta_n(C):=\emptyset$.  Finally for a residue class $R$ we write
\begin{gather} 
  \delta_n(R) :=  \bigcup \left\{ \delta_n(C) : \text{$C\subseteq R$ an
  $n$-active chain} \right\} \, .
\label{def:delta}
\end{gather}
For $(q,x_{\min})$ 
the least element in
$R\cap U_n$ we have that
$\abs{\{ (q,x) \in R\setminus U_n : x_{\min}\leq x\}} \leq
\abs{\delta_n(R)}$.

\begin{example}
  In Figure~\ref{fig:U1} consider the chain~$C:=\{s_4\}\times \{54,63,72,81\}$, which is $1$-active as $(s_4,54)\in U_1$.
  Since $C\setminus U_1=\{s_4\}\times \{72,81\}$  we have that $\delta_1(C)=\{s_4\}\times \{72,75,78,81\}$.
\end{example}

\begin{lemma}
\label{lem: delta(C) not much bigger}
  For all $n \in \NN$ and every chain $C$ we have that
  $|\delta_n(C)|\le |Q|\cdot |C\setminus U_n|$.
\end{lemma}

We now come to the central technical part of the paper, controlling the growth 
of $\delta_n(R)$ as a function of~$n$:

\begin{lemma}\label{lem:claim}
  There exists a polynomial $\poly_2$ such that for each residue class
  $R$ and all $n \in \mathbb{N}$ we have
  \(
    \abs{\delta_{n+1}(R)}  \leq  \max\{\abs{\delta_n( R')} : R' \text{ a residue
    class} \}+ \poly_2(\abs{Q})
  \)
  if $R$ contains a chain that is $(n+1)$-active but not $n$-active.
\end{lemma}

Before proceeding to prove Lemma~\ref{lem:claim}, we demonstrate the underlying intuition.
Consider a configuration $(q,z) \in R\cap U'_{n+1}$ that has a primitive path $\pi$
to a configuration~$(q',z') \in U_{n}$.  
To prove Lemma~\ref{lem:claim}, we argue that $\pi$ lifts to a valid run
from a ``dense'' subset of configurations in $\{(q,z'') \in R : z''\geq z\}$.
There are two main cases in this argument based on whether one of the larger
configurations in the chain induces a valid run ending in a trivial residue
class.

\begin{example}
  The first case occurs in obtaining~$U_1$ from~$U_0$ in the running example;
  see Figure~\ref{fig:U1}. Consider the chain~$C:=\{s_4\}\times \{54,63,72,81\}$.  The
  primitive path~$s_4, s_7, s_8, s_9, s_{10}$ from the largest
  configuration~$(s_4,81)$ in~$C$ leads to a non-trivial $s_{10}$-residue class
  (out of~$U_0$).  However, one among the $n$-next largest configurations
  in~$C$, for $n=\abs{\blocked(s_4, s_7, s_8, s_9, s_{10})} \cdot\abs{Q}$, lifts
  to a valid run to a trivial $s_{10}$-residue class. In the example, this is
  the case for~$(s_4,63)$.
  The second case occurs in obtaining~$U_2$ from~$U_1$ in the running example;
  see Figure~\ref{fig:U2}. Consider the
  chain~$C':=\{s_1\}\times \{12,18,24,\cdots,54\}$.  The primitive path~$s_1, s_3, s_4$,
  from none of the configurations in this chain, ends in a trivial $s_4$-residue
  class.  However, we provide a  subtle argument to  bound $\abs{C'\setminus
  U_2}$  with $\abs{\delta_1(C)}+\poly_2(\abs{Q})$.
\end{example}

\begin{proof}[Proof of Lemma~\ref{lem:claim}]
Pick the minimal element $(q,z_0) \in R\cap U'_{n+1}$.
Moreover, let $(q',z')\in U_n$ and $(q,z_0) \stackrel{\pi}{\rightarrow} (q',z')$
be such that $\pi$ is a shortest run from $(q,z_0)$ to $U_n$.
By Remark~\ref{rmk: Un primitive to Un+1}, $\pi$ is a primitive path.
 
By Proposition~\ref{prop:primitive} there is a run $(q,z_0)
  \stackrel{\pi'}{\rightarrow} (q',z'')$, for some $z'' \geq z'$, such that
  $\pi'$ has length at most $\abs{Q}^2 + 2$, and the residue class $R'$ of
  $(q',z'')$ is either {\em trivial} or the same as the residue class of
  $(q',z')$.
%

  Note that we do not claim that $(q',z'') \in U_n$, nor that $\pi'$ lifts to
  a valid run. In what follows we will argue that if there are more than some
  polynomial number of configurations above $(q,z_0)$ in $C \setminus U'_{n+1}$,
  where $C$ is an $(n+1)$-active chain of $R$, then $\pi'$ does lift to
  a valid run from one of them. Moreover, the run leads to some configuration
  in the same residue class as $(q',z')$ or to a trivial residue class.
  Observe that, intuitively, this means we ``pump'' $\gamma_q$ before taking
  $\pi'$ so if we wanted to reach the same residue class as $(q',z')$ we
  would need some nonnegative integer $c$ such that
  \[
    z_0 + W_q \cdot c + \weight(\pi') \equiv z_0 + \weight(\pi') \pmod{W_{q'}}\, .
  \]
  Based on this intuition, we now identify two cases according to the order of
  $W_q$ in the group $\mathbb{Z}/\mathbb{Z}W_{q'}$ of integers modulo
  $W_{q'}$, which is $\frac{W_{q'}}{\gcd(W_q,W_{q'})}$.  Recall that this
  quantity is the smallest integer $c \geq 1$ such that $W_{q} \cdot c
  \equiv 0 \pmod{W_{q'}}$.

\item \subparagraph{Case~(i):} $\frac{W_{q'}}{\gcd(W_q,W_{q'})} > \abs{Q}$.  We
  first show that
  $\abs{C\setminus U_{n+1}} \leq (\abs{Q}^2 + 2) (\abs{Q}+1)$
  for every
  $(n+1)$-active chain $C$ in $R$.  

  Let $C$ be an $(n+1)$-active chain of $R$ and suppose for a
  contradiction that 
  $\abs{C\setminus U_{n+1}} > (\abs{Q}^2 + 2) (\abs{Q}+1)$.
  Since $C$ is $(n+1)$-active, 
  for every configuration $(q,z) \in C \setminus U_{n+1}$ we have $z \geq z_0$.
  Further, since
  $\pmin(\pi')+z_0 \geq 0$, $\pi'$ can only be
  blocked on 
 a configuration due to a violation of a disequality guard. Since the
  length of $\pi'$ is at most $|Q|^2 + 2$, it follows that at most
  $\abs{Q}^2 + 2$ elements of $C\setminus U_{n+1}$ lie in
  $\{q\} \times \mathrm{blocked}(\pi')$.

  Recall that $C\setminus U_{n+1}$ is upward closed in $C$, so 
  by the assumption that 
  $\abs{C\setminus U_{n+1}} > (\abs{Q}^2 + 2) (\abs{Q}+1)$, there
  exists a set $S:=\{ (q,z_1 + iW_q) : 0 \leq i \leq \abs{Q}\}$ of
  $\abs{Q}+1$ ``consecutive'' elements of $C\setminus U_{n+1}$, for some
  $z_1$, such that no element of $S$ lies in $\{q\} \times 
  \mathrm{blocked}(\pi')$.  Then $\pi'$ lifts to a valid run from each
  element of $S$.  Moreover, since the order of $W_q$ in
  $\mathbb{Z}/\mathbb{Z}W_{q'}$ is assumed to be greater than $\abs{Q}$,
  the images of the elements of $S$, after following $\pi'$, lie in pairwise
  distinct $q'$-residue classes.  But the number of non-trivial $q'$-residue
  classes is at most $\abs{Q}$ and hence some configuration in $S$ has a run
  over $\pi'$ to a trivial $q'$-residue class and hence to $U_n$.  But then such
  a configuration lies in $U_{n+1}$, which is a contradiction. 

  We conclude that
  $\abs{C\setminus U_{n+1}} \leq (\abs{Q}^2 + 2) (\abs{Q}+1)$ for every
  $(n+1)$-active chain $C$ in $R$.  
  But then
$|\delta_{n+1}(C)| \leq \abs{Q} (\abs{Q}^2 + 2) (\abs{Q}+1)$ by Lemma~\ref{lem: delta(C) not much bigger}.
  Finally, since $R$ comprises at most
  $2\abs{Q}$ bounded chains by Remark~\ref{lem:2q}, we have that
$\abs{\delta_{n+1}(R)} \leq 2\abs{Q}^2 (\abs{Q}^2 + 2) (\abs{Q}+1)$.

  \item \subparagraph{Case~(ii):} $\frac{W_{q'}}{\gcd(W_q,W_{q'})} \leq \abs{Q}$.
  For the residue classes $R$ and $R'$ as above, define an injective
  partial mapping $\Phi: \delta_{n+1}(R) \rightarrow \delta_n(R')$ by
  $\Phi(q,x) =(q',x')$ if and only if $x'=x+\mathrm{weight}(\pi')$ and 
  $(q',x') \in \delta_n(R')$.  We will prove that $\Phi$ is defined on
  all but $\poly_3(\abs{Q})$ many configurations in
  $\delta_{n+1}(R)$, for some polynomial $\poly_3$, thereby showing that
  $|\delta_{n+1}(R)| \leq |\delta_n(R')| + \poly_3(\abs{Q})$.  
  To this end,  
  it suffices to show that $\Phi$ is defined on all but
  $\poly_4(\abs{Q})$ many configurations in $\delta_{n+1}(C)$ for
  every $(n+1)$-active chain $C$ in $R$, for some polynomial $\poly_4$.

  Let $C$ be an $(n+1)$-active chain in $R$ and let $C_1,\ldots,C_s$ be
  a list, given in increasing order, of the chains in $R'$ that are
  mapped into by $\Phi$ from some configuration in $\delta_{n+1}(C)$.
  Then $C_1,\ldots,C_s$ are all $n$-active (as they are above $(q',z')\in U_n$).  For $i\in\{1,\ldots,s\}$,
  write $(q,x^{(i)}_{\min})$ for the minimum configuration in
  $\delta_{n+1}(C)$ that is mapped by $\Phi$ to $C_i$ and write
  $(q,x^{(i)}_{\max})$ for the maximum configuration in
  $\delta_{n+1}(C)$ that is mapped to $C_i$.  Then for each
  $i=1,\ldots,s$, every configuration $(q,x) \in \delta_{n+1}(C)$ such
  that $x^{(i)}_{\min} \leq x \leq x^{(i)}_{\max}$
  and
  $x \not\in \times \mathrm{blocked}(\pi')$
  is mapped by $\Phi$ to
  $\delta_n(R')$.  Thus, writing $(q,x_{\max})$ and $(q,x_{\min})$
  respectively for maximum and minimum configurations in
  $\delta_{n+1}(C)$, we have that
  $\Phi$ is defined on all non-blocked
  elements of $\delta_{n+1}(C)$ lying outside the set below.
  %
	\begin{equation}\left\{ (q,x) \in \delta_{n+1}(C) \:\middle|\:
	x \in  \left(x^{(s)}_{\max},x_{\max}\right]  \cup \left[x_{\min},x^{(1)}_{\min} \right) 
	\cup \bigcup_{i=1}^{s-1} \left(x^{(i)}_{\max} ,x^{(i+1)}_{\min}\right) \right\} 
	\label{eq:exception}
	\end{equation}
  Since $\mathrm{blocked}(\pi')$ contains at most $|Q|^2 + 2$ elements,
  it remains to prove that the set~\eqref{eq:exception} has polynomial
  cardinality.
  %
  %
  We claim its size is at most
  $(2\abs{Q}+1)\cdot \poly_5(\abs{Q})$, for some polynomial $\poly_5$.  For this it will suffice to show that
  any sub-interval $I$ of $\delta_{n+1}(C)$ of the form $\{ (q,x) \in
  \delta_{n+1}(C) : a \leq x \leq b\}$, where $a,b \geq x_{\mathrm{min}}$, and
  such that it does not meet the domain of~$\Phi$, has cardinality at most
  $\poly_5(\abs{Q})$. (Indeed, note that~\eqref{eq:exception} is a union of
  at most $2\abs{Q}+1$ such intervals since there are at most $2|Q|$ chains in
  $R$ by Remark~\ref{lem:2q}.)
  
  Let $\poly_6(x) := (x^2 +2)(x+1)+1$.
  Since $\blocked(\pi')$ has cardinality at most
  $|Q|^2 + 2$, if we take
$\poly_6(|Q|)$ consecutive elements of
  $C \setminus U_{n+1}$ then there are at least $\abs{Q}+1$ consecutive
  elements that lie outside $\{q\} \times \blocked(\pi')$ and at least one of these
  elements---say $(q,x)$---has a valid run over $\pi'$ to the residue
  class $R'$ by the assumption that
  $\frac{W_{q'}}{\gcd(W_q,W_{q'})} \leq \abs{Q}$.  Since
  $(q,x)\not\in U_{n+1}$ we have that
  $(q',x+\mathrm{weight}(\pi'))\not\in U_n$ and hence
  $(q,x)$ is in the domain of~$\Phi$. We conclude that any sequence of at least
  $\poly_6(|Q|)$ consecutive elements of $C\setminus U_{n+1}$ meets
  the domain of $\Phi$. Hence any sub-interval 
  $I$, as defined above, contains at most
 $\poly_6(|Q|)$ elements of $C \setminus U_{n+1}$ and, by
    Lemma~\ref{lem: delta(C) not much bigger},
  contains at most $\abs{Q} \cdot \poly_6(\abs{Q})$ elements in total. 
\end{proof}

Proposition~\ref{prop:holes} follows from Lemma~\ref{lem:claim} by induction, as follows.

\begin{proposition}
  There exists a polynomial $\poly_1$ such that for each residue class
  $R$ and all $n\in\mathbb{N}$ we have $|\delta_n(R)| \leq
  \poly_1(\abs{Q})$.
\label{prop:holes}
\end{proposition}
\begin{proof}
  Let $\alpha_n$ be the number of chains in $\Conf_+$ that are $n$-active.
  Since $n$-active chains are by definition bounded, we have that $\alpha_n \leq
  2\abs{Q}^2$ for all $n\in \mathbb{N}$ (see Remark~\ref{lem:2q}).  We argue by
  induction on $n$ that $\abs{\delta_n(R)} \leq \alpha_n \cdot
  \poly_2(\abs{Q})$ for all $n\in \mathbb{N}$ and all residue classes $R$.
  We conclude that $\abs{\delta_n(R)} \leq 2\abs{Q}^2 \cdot
  \poly_2(\abs{Q})$.

The base case is trivial as there are no $0$-active chains and 
  $\delta_0(R)$ is empty for all residue classes.
  The induction step has two cases.  First, suppose that
  $\alpha_{n+1}=\alpha_n$, i.e., all chains in $\Conf_+$ that are $(n+1)$-active
  were already $n$-active.  Since $U_n \subseteq U_{n+1}$, we have that
  $\delta_{n+1}(C) \subseteq \delta_n(C)$ for all chains $C$ in $R$.  We
  conclude that $\delta_{n+1}(R) \subseteq \delta_n(R)$ and so 
  $\abs{\delta_{n+1}(R)} \leq \abs{\delta_n(R)}$. Since $\abs{\delta_n(R)}
  \leq \alpha_n \cdot \poly_2(\abs{Q})$ by induction hypothesis, and $\alpha_n =
  \alpha_{n+1}$ we get that
  $\abs{\delta_{n+1}(R)} \leq\alpha_{n+1} \cdot \poly_2(\abs{Q})$.

  The second case 
  is that $\alpha_{n+1}>\alpha_n$.  Then by Lemma~\ref{lem:claim} we have
  $\abs{\delta_{n+1}(R)} \leq \max\{\abs{\delta_n( R')} : R' \text{ a
  residue class} \}+ \poly_2(\abs{Q})$. Since the right-hand side of the latter
  is at most $\leq \alpha_n \cdot \poly_2(\abs{Q}) + \poly_2(\abs{Q})$, by
  induction hypothesis, and $\alpha_{n+1}>\alpha_n$
  we get that $\abs{\delta_{n+1}(R)} \leq
  \alpha_{n+1} \cdot \poly_2(\abs{Q})$.
\end{proof}

Recall that $(U_n)_{n\in \NN}$ is a monotone sequence. Furthermore, observe that by the proof of Lemma~\ref{lem:claim} the sequence $|\delta_n(R)|$ either strictly decreases, or possible increases if $R$ contains a chain that is $(n+1)$-active but not $n$-active. Since the latter can only take place $|Q|$ times, then $|\delta_n(R)|$ can take a polynomial number of distinct values before converging. Thus, as a consequence of Proposition~\ref{prop:holes} we have:

\begin{corollary}
 The sequence $(U_n)_{n\in\mathbb{N}}$ stabilizes in at most
$\poly_1(\abs{Q})$ steps.
\label{corl:rounds}
\end{corollary}

\subsection{Computing $\Conf_{\infty}$ and Deciding Unboundedness}\label{subsec:algo}
In this section we show how to compute $\Conf_{\infty}$  
in polynomial time and how 
 to decide in polynomial time whether 
the initial configuration $(s,0)$ can reach $\Conf_{\infty}$.

We start by showing that if a configuration can reach $U_n$ via a
  primitive run, then it can also reach $U_n$ via a polynomial-length
  run (see Appendix~\ref{apx: primitive-2 proof} for the proof).  


\begin{proposition}
There exists a polynomial $\poly_7$ such that the following holds. Let $(q,z),(q',z')\in \Conf_+$ and let $(q,z) \stackrel{\pi}{\rightarrow} (q',z')$ 
be a valid run
such that $(q',z') \in U_n$ and $\pi$ is primitive.  Then there is a
valid run $(q,z) \stackrel{\pi'}{\rightarrow} (q',z'')$ such
that $(q',z'') \in U_n$ and $\pi'$ has length at most
$\poly_7(|Q|).$
\label{prop:primitive-2}
\end{proposition}

Recall that $U'_{n+1}$ consists of all configurations in $\Conf_{+}$ with minimal distance to~$U_n$. 
Combining Remark~\ref{rmk: Un primitive to Un+1} and Proposition~\ref{prop:primitive-2}, we have that	the minimal distance from a configuration
$(q,z)\in U'_{n+1}\setminus U_n$ to $U_n$ is at most $\poly_7(|Q|)$.
It follows that we can restrict the
  search for configurations that can reach $U_n$, to those within a
  polynomially-bounded distance to $U_n$.  By itself this is not
  sufficient to obtain a polynomial-time algorithm to decide whether
  $U_n$ is reachable.
   However, using our analysis of the structure of
  $U_n$ in Section~\ref{subsec:struct}, we are able to formulate the
  bounded reachability problem above in a form that admits a
  polynomial-time algorithm.

  Specifically, we consider the \emph{Bounded Coverability problem
    with a Disequality Objective}: Given as input a
  1-VASS~$\VASS=(Q,D,\Delta,w)$ with a distinguished state $q_f$, a
  positive integer~$L$ (written in unary), an initial
  configuration~$(q_0,x_0)$, and a coverability objective of the form
\begin{gather}
  O=\left\{(q_f,x)\mid x\geq \ell \wedge \bigwedge_{i=1}^{m} ({x \not
      \equiv a_i}\mathrel{\mathrm{mod}} W) \wedge \bigwedge_{i=1}^{n}
    (x \neq b_i)\right\},
\label{eq:objective}
\end{gather}
where $\ell,W$ and the $a_i$ and $b_i$ are non-negative integers given
in binary, decide whether $O$ is reachable from $(q_0,x_0)$ via a
valid run of length at most $L$.

\begin{proposition}\label{prop:ptimepath}
The Bounded Coverability problem with a Disequality Objective is decidable in polynomial time.
\end{proposition}

We now show how to compute $\Conf_\infty$ in polynomial	time. 
By Corollary~\ref{corl:rounds},
the sequence~$\{U_n\}_{n\in \mathbb{N}}$ converges
in at most $\mathrm{poly}_1(\abs{Q})$ steps.  It	remains	to show	how to
compute~$U_{n+1}$ from $U_n$	in polynomial time for each $n$.

Recall that all unbounded chains are contained in $U_0$ and hence are
contained in $U_n$ for all $n$.  Recall also that the total number of
bounded chains is at most $2\abs{Q}$ and that $U_n$ is downward closed
in each bounded chain.  Thus $U_n$ is determined by giving, for every
bounded chain $C$ such that $U_n\cap C\neq \emptyset$, the maximum
configuration in $U_n \cap C$.  In particular, $U_n$ can be described
in space polynomial in the description of the given 1-VASS.

Recall that $U_{n+1}$ is obtained from $U_n$ by adding the
configurations in $\mathrm{Conf}_+ \setminus U_n$ that have minimum
distance to $U_n$ and then closing downward in each bounded chain.  By
Remark~\ref{rmk: Un primitive to Un+1} and 
Proposition~\ref{prop:primitive-2}, a configuration in
$\mathrm{Conf}_+ \setminus U_n$ that has minimum distance to $U_n$ has
distance at most $\mathrm{poly}_7(\abs{Q})$.  The idea to compute
$U_{n+1}$ from $U_n$ is as follows:

For each bounded chain $C$, and each configuration $(q,x) \in                                                                            
C\setminus U_n$ that is among the top $\mathrm{poly}_1(\abs{Q})$
configurations in $C$, we determine the distance of $(q,x)$ to $U_n$
up to a bound of $\mathrm{poly}_7(\abs{Q})$.  To do this we use the
procedure described in Proposition~\ref{prop:ptimepath}, having
first written $U_n$ as a polynomial-size union of sets
 of the form~\eqref{eq:objective}---see below for
details.  The reason that it suffices to look only among the top
$\mathrm{poly}_1(\abs{Q})$ configurations in each bounded chain is
because we know from Proposition~\ref{prop:holes} that $\abs{C \setminus U_{n+1}}                                                                       
\leq \mathrm{poly}_1(\abs{Q})$ for every $(n+1)$-active chain $C$.

We next show how to decompose $U_n$ into a polynomial union of
sets of the form~\eqref{eq:objective} in order to apply~Proposition~\ref{prop:ptimepath}.
Fixing $q\in Q_+$, let $R_1,\ldots,R_m$ be a list of the non-trivial $q$-residue classes and for each $i\in\{1,\ldots,m\}$,
write $a_i$ for the corresponding residue modulo~$W_q$ and define $\ell_i:=\min(R_i \cap U_n)$.  Moreover,
let $b_1,\ldots,b_k$ be a list of the counter values such that for all $1\le j\le k$ we have
$b_j\geq \ell_i$ and $(q,b_j)\in R_i\setminus U_n$ for some $i$.
Note that  
$m\leq \abs{Q}$ and $k\leq m\,\mathrm{poly}_1(\abs{Q})$, and the corresponding classes and numbers can be enumerated in polynomial time.
We  decompose the set of configurations $\{
(q,z) \in U_n\}$ into the following two components:
\begin{enumerate}
\item $\{(q,z): z\ge \pmin(\gamma_q)\wedge \bigwedge_{i=1}^m z\not\equiv a_i\pmod{W_q}\}$, i.e., all configurations in trivial $q$-residue classes,
\item for all 
$j \in \{1,\ldots,m\}$, the set $\{(q,z): z\ge \ell_j \wedge \bigwedge_{i: i\neq j} z\not\equiv a_i\pmod{W_q} \wedge \bigwedge_{i=1}^k z\neq b_i\}$, which includes
$R_j \cap U_n$ for  the non-trivial residue class~$R_j$.
\end{enumerate}


Finally, it remains to decide whether the configuration $(s,0)$ is
unbounded.  By Proposition~\ref{prop: unbounded iff confplus unbounded}, 
$(s,0)$ is unbounded if and only if it can reach
$\Conf_\infty$.  Now a shortest run from $(s,0)$ to $\Conf_\infty$ is
necessarily primitive: if an internal configuration in such a run lies
in $\Conf_+$ then it is also in $\Conf_{\infty}$---a contradiction.  By
Proposition~\ref{prop:primitive-2}, a shortest run from $(s,0)$ to
$\Conf_\infty$ has length at most $\mathrm{poly}_7(\abs{Q})$.  Thus we
can decide whether such a run exists in polynomial time using
Proposition~\ref{prop:ptimepath}.  In conclusion we have

	\begin{theorem}\label{thm:unboundedness-ptime}
		The Unboundedness Problem and the Coverability Problem for 1-VASS with disequality tests
		are decidable in polynomial time.
	\end{theorem}

\section{Unboundedness for 1-VASS}\label{sec:1vass}
In this section we show that the Unboundedness Problem for 1-VASS (i.e., with no
disequality tests) is in $\NC^2$.
 Recall that $\NC^i$ is the class of decision
problems solvable in time $O(\log^i n)$, with $n$ the size of the input, on a
parallel computer with a polynomial number of
processors~\cite{papadimitriou94,ab09}.

Let $\VASS=(Q,\Delta,w)$ be a 1-VASS with a distinguished state $s \in Q$.  We
want to decide whether the configuration $(s,0)$ is unbounded.  Since $\VASS$
has no disequality tests, deleting a negative-weight or zero-weight cycle that
appears as an infix of a valid run yields another valid run.  It follows that
$(s,0)$ is unbounded if and only if there is a valid run from $(s,0)$ consisting
of a simple path (of length at most $\abs{Q}$) followed by a positive-weight
simple cycle (again, of length at most $\abs{Q}$).  We call such a run a
\emph{lasso}.

Let $\VASS=(Q,\Delta,w)$ be a 1-VASS and let $\pi=q_1,\ldots,q_n$ be a path
in $\VASS$. Recall that  a (possibly empty) prefix of
$\pi$ is said to be \emph{minimal} if it has minimal weight among all
prefixes of $\pi$.  Likewise a (possibly empty) suffix of $\pi$ is
said to be \emph{maximal} if it has maximal weight among all suffixes.
It is clear that $q_1,\ldots,q_m$ is a minimal prefix of $\pi$ if and
only if $q_m,\ldots,q_n$ is a maximal suffix.  In such a case let us
call $q_m$ a \emph{nadir} of $\pi$ (the nadir is the lowest point
reached in any run over $\pi$).  Recall that $\mathrm{pmin}(\pi)$ is the
weight of a minimal prefix of $\pi$; correspondingly we  define $\mathrm{smax}(\pi)$ to be the
weight of a maximal suffix.

Given paths $\pi$ and $\pi'$, say that $\pi$ is \emph{dominated} by
$\pi'$ if $\mathrm{pmin}(\pi) \leq \mathrm{pmin}(\pi')$ and
$\mathrm{smax}(\pi) \leq \mathrm{smax}(\pi')$. 
Observe that if $\pi$ is dominated by~$\pi'$ then 
$\weight(\pi)\leq \weight(\pi')$.
\begin{example}
  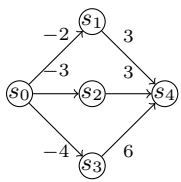
\begin{figure}[ht]
    \centering
    \begin{tikzpicture}[inner sep=2mm, auto, node distance=0.6cm, every
      state/.style={font=\small,inner sep=0, minimum
      size=1mm,},el/.style={font=\scriptsize}]
      \node[state](s0){$s_0$};
      \node[state,right= of s0](s2){$s_2$};
      \node[state,above= of s2](s1){$s_1$};
      \node[state,below= of s2](s3){$s_3$};
      \node[state,right= of s2](s4){$s_4$};
      \path[->]
        (s0) edge node[el,above]{$-2$} (s1)
        (s0) edge node[el,above]{$-3$} (s2)
        (s0) edge node[el,below]{$-4$} (s3)
        (s1) edge node[el,above]{$3$} (s4)
        (s2) edge node[el,above]{$3$} (s4)
        (s3) edge node[el,below]{$6$} (s4)
      ;
    \end{tikzpicture}
    \caption{The topmost path dominates the middle one; the bottom path
    dominates no other path}
    \label{fig:domination-example}
  \end{figure}
  In Figure~\ref{fig:domination-example}, the path $s_0,s_1,s_4$ dominates $s_0,s_2,s_4$. However, despite it being
  the case that $\weight(s_0,s_3,s_4) > \weight(s_0,s_2,s_4)$, $s_0,s_3,s_4$
  does not dominate $s_0,s_2,s_4$ since the weight of a minimal prefix of the former is
  smaller than that of the latter.
\end{example}
Fix two states
$p,q \in Q$ and let $P$ be a set of $p$-$q$ paths.  We say that a set
$P'$ of $p$-$q$ paths is a \emph{Pareto set} for $P$ if for every
$\pi \in P$ there exists $\pi' \in P'$ such that $\pi$ is dominated by
$\pi'$.

We observe some simple properties of Pareto sets:
\begin{lemma}
Let $p,q,r \in Q$.  Then all of the following statements hold:
\begin{enumerate}
\item If $P_1,P_2,P_3$ are sets of $p$-$q$ paths such that $P_1$ is a Pareto set of
  $P_2$ and $P_2$ is a Pareto set of $P_3$, then $P_1$ is a Pareto set of
  $P_3$.
\item If $P,R$ are sets of $p$-$q$ paths with
  respective Pareto sets $P',R'$, then
  $P'\cup R'$ is a Pareto set for
  $P\cup R$
\item If $P$ is a set of $p$-$q$ paths and $R$ is
  a set of $q$-$r$ paths with respective Pareto sets
  $P',R'$, then $P'\cdot R'$
  is a Pareto set of $P\cdot R$.
\end{enumerate}
\label{lem:simple}
\end{lemma}

\begin{proposition}
  Let $p,q \in Q$.  Then every set $P$ of $p$-$q$ paths of length at
  most $k$ has a Pareto set $P'$ of cardinality at most $|Q|$ such
  that each path in $P'$ has length at most $2k$.  Moreover such a set
  $P'$ can be computed from $P$ in $\NC^1$.
\label{prop:filter}
\end{proposition}

\subsection*{An $\NC^2$ Upper Bound}

\begin{theorem}\label{thm:nc2upper}
The Unboundedness Problem and the Coverability Problem for
  1-VASS are  decidable in $\NC^2$.
\end{theorem}
\begin{proof}
  By Lemma~\ref{lem:cov2unbound}, it will suffice to show that Unboundedness is in $\NC^2$.

  Let $\VASS=(Q,\Delta,w)$ be a 1-VASS.  Given $p,q \in Q$ and
  $m \in \mathbb{N}$, denote by $\mathrm{Paths}_{p,q,m}$ the set of all
  $p$-$q$ paths in $\mathcal{V}$ of length at most $m$.  

  Given a state $s\in Q$, recall that
  $(s,0)$ is unbounded if and only if there exists a lasso run that
  starts at $(s,0)$. 
  To determine the existence of such a run we
  compute a Pareto set $P_q$ for $\mathrm{Paths}_{s,q,\abs{Q}}$ and a
  Pareto set $P'_q$ for $\mathrm{Paths}_{q,q,\abs{Q}}$ for every state
  $q \in Q$.  Having done this we look for $q\in Q$ and paths
  $\pi \in P_q$ and $\pi' \in P'_q$ such that $\pi\cdot \pi'$ induces a valid
  run from $(s,0)$ and $\pi'$ has positive weight.


  It remains to show how to compute a Pareto set of
  $\mathrm{Paths}_{p,q,\abs{Q}}$  for all pairs of states~$p,q\in Q$
  (together with the values $\mathrm{weight}(\pi)$ and
  $\pmin(\pi)$ for every path $\pi$ in the Pareto set) in
  $\NC^2$.

  For $k=1,\ldots,\lceil \log |Q| \rceil$, we show how to compute a
  family $\mathcal{P}_k=\{ P_{p,q,k} \}_{p,q \in Q}$ such that for all~$p,q \in Q$: 
  \begin{enumerate}
  \item $P_{p,q,k}$ is a Pareto set for $\mathrm{Paths}_{p,q,2^k}$;
  \item $P_{p,q,k} \subseteq \mathrm{Paths}_{p,q,4^k}$;
  \item $|P_{p,q,k}|\leq |Q|$.
  \end{enumerate}
  By Item 1, if $k= \lceil \log |Q| \rceil$ then
  $P_{p,q,k}$ is a Pareto set for $\mathrm{Paths}_{p,q,|Q|}$.
  (Note that for $k=
  \lceil \log |Q| \rceil$, $\mathcal{P}_k$ consists of paths of length
  at most $|Q|^2$.)

  The construction of $\mathcal{P}_k$ is by induction on $k$.  Suppose
  we have computed $\mathcal{P}_k$ with Properties 1-3~above.  Fix
  $p,q \in Q$.  In order to compute $P_{p,q,k+1}$ we observe that
  \begin{equation}
    P:=\{ \pi_1 \cdot \pi_2 : \exists r\in Q (\pi_1 \in P_{p,r,k} \wedge
    \pi_2 \in P_{r,q,k}) \}
  \label{eq:setP}
  \end{equation}
  is a Pareto set for $\mathrm{Paths}_{p,q,2^{k+1}}$ by Items 2 and 3 of
  Lemma~\ref{lem:simple}.
  Applying
  Proposition~\ref{prop:filter}, we obtain a Pareto set $P'$ for $P$ of
  cardinality at most $|Q|$. 
  By Item 1 of
  Lemma~\ref{lem:simple}, $P'$ is a Pareto set for
  $\mathrm{Paths}_{p,q,2^{k+1}}$.  Finally, it is clear from the length
  bound in Proposition~\ref{prop:filter} that all paths in $P'$ have
  length at most $4^{k+1}$.  Thus we define $P_{p,q,k+1}:=P'$.

  It remains to establish the $\NC^2$ complexity bound
  for computing $\mathcal{P}_{\lceil \log |Q|
  \rceil}$.  For this it suffices to show that for all $k$ the computation of
  $\mathcal{P}_{k+1}$ from $\mathcal{P}_k$ can be carried out in
  $\NC^1$.  But we may compute each set $P_{p,q,{k+1}}$ in
  parallel (over $p,q\in Q$), and the computation of each such set can
  be done in $\NC^1$ by Proposition~\ref{prop:filter}.
\end{proof}

\section{Conclusion}
We have shown that control-state reachability for 1-VASS with
disequality tests can be solved in polynomial time.  The complexity of
reaching a given \emph{configuration} in this model is open (being
equivalent to control-state reachability in the presence of both
equality and disequality tests), lying between NP and PSPACE.  For
multi-dimensional VASS with disequality tests, the classical argument
of Rackoff~\cite{Rackoff78} easily generalises to show that
control-state reachability remains in EXPSPACE.  By contrast,
decidability of reachability is open to the best of our knowledge.
For comparison, recall that without disequality tests reachability is
decidable but non-elementary~\cite{CzerwinskiLLLM19}.

\clearpage
\bibliography{refs}

\clearpage
\appendix
\section{Proof of the reduction in Figure~\ref{fig:GCD}}
\label{apx:gcd}
Let us recall that for every value $u\in \NN$, the  assignment $\mathrm{val}_u:\set{X_1,\ldots,X_m}\to \set{0,1}$ is defined  by $\mathrm{val}_u(X_i)=1$ if
 and only if $p_i \mid u$.
For convenience,  define the domain $D_s\subseteq \mathbb{N}$ containing  all allowable counter values in state $s$ (exclude all disequality guards on~$s$).

The key observation is the following: let $u\in \{0,\ldots,P-1\}$, and consider a clause $C_i=\ell_{i_1}\vee \ell_{i_2}\vee \ell_{i_3}$, where $\ell_{i_j}$ is a literal of variable $X_{i_j}$, then $\mathrm{val}_u$ satisfies $C_i$ iff there exists some $k\in \NN$ such that $u+k p_{i_1}p_{i_2}p_{i_3}\notin D_i$. 

Indeed, note that for every $j\in \set{1,2,3}$ and every $k\in \NN$ we have that $p_{i_j}|u$ iff $p_{i_j}|u+k p_{i_1}p_{i_2}p_{i_3}$.  
Recall that ${\rm val}_u(X_{i_j})=1$ iff $p_{i_j}|u$, and observe that since $u<P$, there exists $k\in \NN$ such that $u+k p_{i_1}p_{i_2}p_{i_3}\in \{P,P+1,\ldots,P+p_{i_1}p_{i_2}p_{i_3}-1\}$. We thus have that $\mathrm{val}_u$ satisfies $C_i$ iff $\mathrm{val}_{u+k p_{i_1}p_{i_2}p_{i_3}}$ satisfies $C_i$, iff $u+k p_{i_1}p_{i_2}p_{i_3}\notin D_{i}$.

Now, assume $\varphi$ is satisfiable, and let $\pi$ be a satisfying assignment. We associate with $\pi$ the number $u=\prod_{j: \pi(X_j)=1}p_j \pmod{P}$ (note that taking modulo $P$ simply means that if the product is exactly $P$, we take $u=0$). Clearly $\pi=\mathrm{val}_u$. We claim that $(s_0,u)$ is bounded. Indeed, the only paths possible from $(s_0,u)$ start by choosing a state $s_i$, and then repeatedly applying the cycle of cost $c_i$. However, since $\mathrm{val}_u$ satisfies all clauses, then by the above, all such paths are blocked by a disequality guard after taking the $c_i$ for $k$ times, for some $k\in \NN$ (which depends on $i$). Thus, $(s_0,u)$ is bounded.

Conversely, assume $(s_0,u)$ is bounded for some value $u$, we claim that $\mathrm{val}_u$ satisfies $\varphi$. Indeed, by the same reasoning above, it follows that for every cycle of cost $c_i$, we have $u+k c_i\notin D_{i}$ for some $k\in \NN$, so $\mathrm{val}_u$ satisfies $C_i$. Since this is true for all clauses, we have that $\mathrm{val}_u$ satisfies $\varphi$.

We conclude that $\varphi$ is satisfiable iff some configuration $(s_0,u)$ is bounded, which completes the reduction.

Finally, we note that the reduction indeed takes polynomial time --- indeed, the
construction clearly has polynomially many states. Also, the first $m$ primes
$p_1,\ldots, p_m$ can be listed in time polynomial in $m$, and are representable
in polynomially many bits. Therefore, the binary representation of the
transition values and the amount of missing elements in the domain of each state
are both polynomial.

\section{Single disequality guards suffice}\label{apx:guard}
Given a 1-VASS $\VASS=(Q,\Delta,D,w)$ with disequality tests, we can assume  that for all states~$q$ the set $D_q$ is
either~$\mathbb{N}$ or $\mathbb{N}\setminus \{g\}$ for some $g\in\mathbb{N}$.
This assumption is without loss of generality, as  
a state~$q$ with $D_q = \mathbb{N}\setminus \{a_1,\ldots,a_n\}$ 
can be replaced
with a sequence of new states $q_1,\cdots,q_n$, connected with $0$-weight transitions,
such that  $D_{q_i}=\mathbb{N}\setminus \{a_i\}$ for $i\in\{1,\ldots,n\}$.
The transformation yields only a polynomial blow-up in the size of the 1-VASS,
and there is a natural correspondence between runs in the original 1-VASS and
the modified one.

\section{Proof of Lemma~\ref{lem:cov2unbound}}
\label{apx: reduction coverability to unboundedness}
Consider a 1-VASS $\VASS=(Q,\Delta,D,w)$ with disequality tests, and let $s,t\in
Q$. We reduce the Coverability problem to the Unboundedness problem as follows.

  We obtain from $\VASS$ a new 1-VASS $\VASS'$ as follows. First, we remove from
  $\VASS$ all the states that cannot reach $t$ in the underlying graph. Second, we
  introduce a new state $t'$ with a self-loop of weight $+1$, that is reachable
  from $t$ with a transition of weight $0$. The output of the reduction is
  $\VASS'$ with the distinguished state $s$. 

Recall that reachability in directed graphs can be decided in $\NL
  \subseteq \NC^2$, and hence this reduction is $\NC^2$-computable.

  Henceforth assume that $s$ can reach $t$ in the underlying graph of $\VASS$
  (otherwise $s$ cannot cover $t$, and the reduction can output a trivial
  negative instance).  We proceed to prove the correctness of the reduction.

  First, if $(s,0)$ can cover $t$ in $\VASS$, then in particular it can only
  cover $t$ using states in $\VASS'$.  We now have that $(s,0)$ is unbounded in
  $\VASS'$, by covering $t$, and then taking the transition to $t'$ and
  repeating the self loop unboundedly. Note that crucially, there are no
  disequality guards on $t'$, and therefore once $t$ is reached, we can take the
  transition to $t$ and repeat the self loop unboundedly. 

  Conversely, suppose $(s,0)$ is unbounded in $\VASS'$, then either there is a
  valid run in $\VASS$ from $(s,0)$ to $(t',z)$ for some $z$, in which case
  $(s,0)$ can cover $t$ in $\VASS$, or $(s,0)$ is unbounded already in $\VASS$
  and, moreover, it is unbounded in $\VASS$ using only states that can reach $t$
  in the underlying graph. We claim that in the latter case, $(s,0)$ can cover
  $t$ in $\VASS$. Indeed, from $(s,0)$ there is a valid run to a configuration
  $(q,z)$ with $z$ that is large enough, such that a simple path from $q$ to $t$
  in the underlying graph lifts to a valid run from $(q,z)$ to $(t,z')$ for some
  $z'$. Specifically, taking $z>|Q|\cdot W\cdot G$ where $W$ is the maximal
  absolute value of the weight of a transition in $\VASS$, and $G$ is the
  maximal disequality guard, suffices for such a run.

\section{Proof of Proposition~\ref{prop: unbounded iff confplus unbounded}}

Clearly if $(s,0)$ can reach an unbounded configuration in  $\Conf_+$ then it is unbounded.

Conversely, if $(s,0)$ is unbounded, then 
there is a state~$q$ such that for all $z_0\in \mathbb{N}$,
there exist $z,z'\geq z_0$ and  a 
valid  run~$\pi$ starting in $(s,0)$ 
that visits $(q,z)$ and ends in~$(q,z')$.
Thus, there is a positive cycle~$\gamma$ on~$q$.  
The positive cycle~$\gamma$ on~$q$ may not be simple, but 
it certainly visits a state~$p$ with a simple positive cycle~$\gamma_p$
on it.
Pick  $z_0$ such that $z_0>\pmin(\gamma)+x$. 
for all $x\in \blocked(\gamma_p^{\omega})$
(Note that
$\blocked(\gamma_p^{\omega})$ is finite since $\gamma_p$ is a positive cycle.
The maximum is thus well-defined.)
Hence, there is a valid run from $(s,0)$ to $(p,y)$
where $y> \max(\blocked(\gamma_p^{\omega}))$.
Observe that $(p,y)\in \Conf_+$ and it   is unbounded.

\section{Proof of Proposition~\ref{prop:primitive}}
Suppose that $\pi$ has length strictly greater than $\abs{Q}^2+2$.  
By the Pigeonhole principle, we can find $\abs{Q}+1$ distinct proper prefixes (i.e. prefixes that are not just the initial state, or the entire path) of~$\pi$ that end in the
same state. That is, $\abs{Q}$ proper cycles on the same state.  Let $\pi_1,\ldots,\pi_{\abs{Q}+1}$ be a list of
these prefixes, given in order of increasing length, and let the
corresponding suffixes be $\pi'_1,\ldots,\pi'_{\abs{Q}+1}$.  We now
consider two cases.

  First,
  suppose that there exist $i<j$ such that $\mathrm{weight}(\pi_i)$ and
  $\mathrm{weight}(\pi_j)$ have the same residue modulo $W_{q'}$.  Then
  define $\pi':= \pi_i \cdot \pi'_j$.  In this case path 
  $\pi'$ lifts to
  a run from $(q,z)$ to $(q',z'')$ such that $(q',z'')$ lies in
  the same $q'$-residue class as $(q',z')$.  
  The second case is
  that the respective residue classes of $\mathrm{weight}(\pi_1),\ldots,
  \mathrm{weight}(\pi_{\abs{Q}+1})$ modulo~$W_{q'}$ are all distinct. 
  Then there exists $i>1$ such that, defining $\pi':= \pi_1 \cdot \pi'_i$,
  the path $\pi'$ lifts to a run from $(q,z)$ to $(q',z'')$ such
  that $(q',z'')$ lies in a trivial $q'$-residue class (as there are at most
  $|Q|$ non-trivial residue classes).
  
  Continuing in this fashion we can recursively remove cycles from the
  original path $\pi$ to eventually obtain a path $\pi'$ that has length
  at most $\abs{Q}^2 + 2$ and such that Item 3 is satisfied. 
  Consider all
  maximal infixes that were removed from $\pi$ to obtain $\pi'$. Note that each
  such infix must necessarily be a cycle as they arise from iteratively removing
  cycles. Since $\pi$ was primitive, all of them must have
  non-positive weight. Hence, Items 1 and 2 also hold\footnote{Note that we do not claim that the intermediate paths obtained in the  procedure  are primitive nor that the individual cycles removed in this process are negative. Rather the observation is  that $\pi'$ can equivalently be obtained from $\pi$ in one step by simultaneously removing a disjoint family of infixes, where each infix is a cycle (necessarily non-positive).}.

\section{Proof of Lemma~\ref{lem: delta(C) not much bigger}}

  Consider two ``consecutive'' configurations $(q,z),(q,z+W_q)\in C\setminus
  U_n$, then all configurations $(q,z')$ for $z\le z'<z+W_q$ lie in
  pairwise-distinct $q$-residue classes. In particular, since there are at most
  $|Q|$ non-trivial residue classes, and since trivial residue classes are
  contained in $U_0$, we have that at most $|Q|$ such elements are in
  $\delta_n(C)$.

\section{Proof of Proposition~\ref{prop:primitive-2}}
\label{apx: primitive-2 proof}
By Proposition~\ref{prop:holes} we can find a polynomial $\mathrm{poly'}_7$ such that
\begin{gather}
  \mathrm{poly'}_7(\abs{Q}) \geq \abs{Q}^2 + \abs{Q} + 3 + \sum_{R \text{
  non-trivial}} |\delta_n(R)| 
\label{eq:poly-bound}
\end{gather} for all $n \in \mathbb{N}$.

Set $\poly_7(|Q|):=|Q|\cdot (\poly'_7(|Q|))^2+|Q|^2 + 4$, 
and consider a valid, primitive path $\pi$ such that 
$\mathrm{length}(\pi) > \poly_7(|Q|)$ 
and $(q,z)\stackrel{\pi}{\to} (q',z')$.

Since $\pi$ has length greater than~$|Q| \cdot (\poly'_7(|Q|))^2+2$, there
exists a state $q''\in Q$ that occurs at least $(\poly'_7(|Q|))^2$ times in internal
configurations within the first $|Q| \cdot (\poly'_7(|Q|))^2+2$ configurations
of $\pi$.
Thus, there exists a sequence of proper prefixes $\pi_1 < \ldots <
\pi_{\poly'_7(|Q|)}$ of $\pi$ that all end in~$q''$ and such that one of the
following two cases holds.
\begin{enumerate}[(i)]
  \item The numbers $\mathrm{weight}(\pi_i)$ all have the same residue
    modulo~$W_{q'}$.  
  \item The numbers $\mathrm{weight}(\pi_i)$  have pairwise distinct residues
    modulo~$W_{q'}$.
\end{enumerate}
Indeed, since there are $(\poly'_7(|Q|))^2$ prefixes to choose from, either Case
(i) holds, or there are strictly less than $\poly'_7(|Q|)$ prefixes per residue
class. If the latter holds then there must be least $\poly'_7(|Q|)$ such
distinct residue classes, so Case (ii) holds.

In either case, we decompose the computation $\pi$ as $\pi = \pi_{\poly'_7(|Q|)}
\cdot \pi'$. Observe that since $\pi$ is primitive, then
so is $\pi'$.  Applying Proposition~\ref{prop:primitive} to $\pi'$ we obtain a
path $\pi''$ of length at most $|Q|^2 + 1$ such that $\pi_{\poly'_7(|Q|)} \cdot
\pi''$ leads from $(q,x)$ to either the same residue class as $(q',z')$ or to a
trivial $q'$-residue class.  

It is important to note that we cannot assume $\pi''$ is not blocked after the prefix
$\pi_{\poly'_7(|Q|)}$. However, since $|\blocked(\pi'')|\le |Q|^2$, we can
remove from the list of prefixes at most $|Q|^2$ prefixes such that the
remaining prefixes do not cause $\pi''$ to block. (Indeed, we will not modify
the path by literally removing prefixes but rather cycles which correspond to the
path from a prefix to a longer prefix. For now, we are only speaking about
removing elements from the collection of prefixes we can choose from.) W.l.o.g,
let $\pi_1,\ldots \pi_{d}$ be the remaining prefixes.


Consider the family of paths $\theta_i:=\pi_i \cdot \pi''$
for $i\in\{1,\ldots,d\}$. Note that every $\theta_i$ is of length at most
$\poly_7(|Q|)$, and since the $\theta_i$ are obtained by removing $q''$-cycles,
and since $\pi$ is primitive, the configurations reached by $\theta_i$ are
above $(q',z')$.  We claim that one of the $\theta_i$ is a valid run from
$(q,z)$ to $U_n$.

We separate the analysis according to the cases above. 
\begin{itemize}
  \item In Case (i), if $\pi''$ leads to a trivial residue class, then all the
    $\theta_i$ reach $U_n$, and we are done. Otherwise, $\pi''$ leads to the
    same residue class as $(q',z')$. By our choice of $\poly'_7(|Q|)$
    in~\eqref{eq:poly-bound}, we have that $d>\sum_{R \text{
    non-trivial}}|\delta_n(R)|$.
    That is, there are more prefixes that do not cause $\pi''$ to block than
    there are missing elements above $(q',z')$ in $U_n$. We conclude that some
    $\theta_i$ reaches $U_n$.
  \item In Case (ii), the paths $\theta_i$ all reach distinct residue classes.
    In particular, since there are more than $|Q|$ such prefixes --- i.e.  $d >
    |Q|$ by our choice of $\poly'_7(|Q|)$ --- then some $\theta_i$ reach trivial
    residue classes, and thus reach $U_n$.
\end{itemize}
%

\section{Proof of Proposition~\ref{prop:ptimepath}}
  We carry out a forward reachability analysis starting from the
  initial configuration~$(q_0,x_0)$.   The algorithm runs for
    $L+1$ rounds.  In the $k$-th round, we maintain for each state~$q$
    a set~$S_{q,k}$ of configurations~$(q,x)$ that are reachable
    from~$(q_0,x_0)$ by valid runs of length~$k$.  Let $R_{q,k}$
    denote the set of all configurations $(q,x)$ that are reachable
    from~$(q_0,x_0)$ by valid runs of length~$k$.  We maintain the
    invariant that if some configuration $(q,x) \in R_{q,k}$ can reach
    the objective $O$ in $L-k$ steps via a path $\pi$ then some
    configuration $(q,x') \in S_{q,k}$ can also reach $O$ via the same
    path $\pi$.  We output that the objective is reachable if and only
    if one of the sets $S_{q_f,k}$ for some $k\in \{0,\ldots,L\}$
    intersects $O$.  This last step is clearly sound, given the
    invariant.

The key to obtaining a polynomial-time runtime bound is to suitably prune the
sets~$S_{q,k}$ to keep them of polynomial size.  In order to compute
$\{S_{q,{k+1}}\}_{q\in Q}$ from $\{ S_{q,k} \}_{q\in Q}$ we proceed as
follows.  First define $\{ S'_{q,k} \}_{q\in Q}$ to be the indexed set
of all valid configurations reachable in one step from
$\{ S_{q,k} \}_{q\in Q}$.  Now we obtain $S_{q,k+1}$ from $S'_{q,k}$
by the following two steps:
\begin{itemize}
	\item First, we delete from~$S'_{q,k}$ all configurations~$(q,x)$ such that
	there are $(n+L)$ configurations~$(q,x')$ in $S'_{q,k}$ with
	$x'> x$ and $x' \equiv x\pmod{W}$. 
	\item Secondly,  we delete from~$S'_{q,k}$ all configurations~$(q,x)$ such that
	there are $(n+L)(m+1)$ configurations~$(q,x')$ in $S'_{q,k}$ with $x'>x$.
\end{itemize}

Clearly each set $S_{q,k}$ has cardinality at most~$(n+L)(m+1)$, and
moreover, it can be computed from the collection of sets
$\{S_{q',k-1} \mid q'\in Q\}$ in polynomial time.

It remains to argue that the invariant is maintained between rounds.
To this end, suppose some state $(q,x) \in R_{q,k+1}$ can reach the
objective in $L-k-1$ steps via a path $\pi$.  Then there exists a
state $(q',x') \in R_{q',k}$ that can reach the objective in $L-k$
steps via the path $q'\pi$.  By the loop invariant there exists a state
$(q',x'') \in S_{q',k}$ that can also reach the objective via the path
$q'\pi$.  Hence there is a state $(q,y) \in S'_{q',k}$ that can reach
the objective via the path $\pi$.  Now if $(q,y)$ is deleted in the first
stage of pruning then there is some configuration $(q,y')$ such that
$y'>y$, $y' \equiv y \pmod{W}$, and $\pi$ yields a valid computation
from $(q,y')$ to the objective $O$.  After the first stage of pruning,
each residue class in $S'_{q,k}$ contains at most $n+L$ elements.
Hence if $(q,y')$ is deleted in the second stage of pruning, there are
at least $n+L$ configurations $(q,y'')$ in $S_{q,k+1}$ that are above
$(q,y')$ and are such that the run over $\pi$ from $(q,y')$ leads to a
configuration $(q_f,z)$ with
$\bigwedge_{i=1}^m {z \not\equiv a_i} \bmod W$.  Now from one of these
configurations $\pi$ yields a valid run that reaches $O$ since 
one of $n+L$ choices of $(q,y'')$ will avoid $\mathrm{blocked}(\pi)$ and lead
to a configuration $(q_f,z)$ such that $\bigwedge_{i=1}^n z\neq b_i$.

\section{Proof of Lemma~\ref{lem:simple}}
  Items 1 and 2 are obvious.  Item 3 follows from the fact that if
  $\pi_1 \in P$ is dominated by $\pi'_1 \in P'$ and
  $\pi_2 \in R$ is dominated by $\pi'_2 \in R'$ then
  $\pi_1 \cdot \pi_2$ is dominated by $\pi_1' \cdot \pi_2'$.  Indeed,
\begin{eqnarray*}
\mathrm{pmin}(\pi_1 \cdot \pi_2) & = &
\min(\mathrm{pmin}(\pi_1),\mathrm{weight}(\pi_1)+\mathrm{pmin}(\pi_2))\\
& \leq &
\min(\mathrm{pmin}(\pi'_1),\mathrm{weight}(\pi'_1)+\mathrm{pmin}(\pi'_2))\\
& = & \mathrm{pmin}(\pi'_1 \cdot \pi'_2) \, .
\end{eqnarray*}
We can similarly argue that
$\mathrm{smax}(\pi_1 \cdot \pi_2) \leq \mathrm{smax}(\pi'_1 \cdot
\pi'_2)$.

\section{Proof of Proposition~\ref{prop:filter}}

  Fix a state $r\in Q$. Consider all $p$-$r$ paths that appear as a
  minimal prefix of some path in $P$.  Pick a single such prefix
  $\pi_1$ of maximum weight.  Likewise consider all $r$-$q$ paths that
  appear as a maximal suffix of some path in $P$ and pick a single
  such suffix $\pi_2$ of maximum weight.  Now form the path
  $\pi:=\pi_1 \cdot \pi_2$.  This path dominates any path in $P$ with
  nadir $r$.  We define $P'$ to be the set of paths $\pi$ formed in
  this way as $r$ runs through $Q$. 
  By taking $k$ large enough, we can suppose without loss of generality, that the absolute weight of all paths in $P'$ is at most $2^k$.
  That is, it can be encoded in binary using $k+1$ bits.

  The $\NC^1$ bound on computing $P'$ relies on the well-known
  fact that the sum of a list of binary integers can be computed in
  $\NC^1$~\cite[Chapter 1]{vollmer99}. To obtain $P'$ we compute the
  weight of each prefix and suffix of every path in $P$ in parallel.
According to~\cite{vollmer99}, this can be done in time $O(\log k)$
  on a parallel computer with $|P|k$ processors: one for each element of $P$ and
  each midpoint $0 \leq m \leq k$. Finally, for each state $r\in Q$ in parallel, we
  find a maximum-weight prefix of a path in $P$ that connects $p$ and $r$ and a
  maximum-weight suffix of a path in $P$ that connects $r$ and $q$. It is
  straightforward to prove the latter is also in $\NC^1$ since sorting a list
  of numbers can be done in $\NC^1$,~\cite{preparata78,bc94}
  thus completing the proof.


\end{document}